\documentclass[aps,12pt,twoside,notitlepage,superscriptaddress]{revtex4-1}

\usepackage{graphicx,epic,eepic,epsfig,amsmath,amsfonts,latexsym,amssymb,color,xcolor,enumerate}
\usepackage{bm,bbm}
\allowdisplaybreaks[3]

\usepackage[colorlinks,
            linkcolor=blue,
            anchorcolor=blue,
            citecolor=red,
            ]{hyperref}

\def\squareforqed{\hbox{\rlap{$\sqcap$}$\sqcup$}}
\def\qed{\ifmmode\squareforqed\else{\unskip\nobreak\hfil
\penalty50\hskip1em\null\nobreak\hfil\squareforqed
\parfillskip=0pt\finalhyphendemerits=0\endgraf}\fi}
\def\endenv{\ifmmode\;\else{\unskip\nobreak\hfil
\penalty50\hskip1em\null\nobreak\hfil\;
\parfillskip=0pt\finalhyphendemerits=0\endgraf}\fi}

\newtheorem{theorem}{Theorem}

\newtheorem{definition}[theorem]{Definition}

\newtheorem{lemma}[theorem]{Lemma}

\newtheorem{proposition}[theorem]{Proposition}
\newtheorem{remark}[theorem]{Remark}

\newenvironment{proof}[1][Proof]{\noindent\textbf{Proof.} }{\hfill\qed}
\newenvironment{proofof}[1][Proof]{\noindent\textbf{Proof~#1.} }{\hfill\qed}

\newcommand{\nc}{\newcommand}
\nc{\rnc}{\renewcommand}
\nc{\beq}{\begin{equation}}
\nc{\eeq}{\end{equation}}
\nc{\bsp}{\begin{split}}
\nc{\esp}{\end{split}}
\nc{\beqa}{\begin{eqnarray}}
\nc{\eeqa}{\end{eqnarray}}
\nc{\lbar}[1]{\overline{#1}}
\nc{\ket}[1]{|#1\rangle}
\nc{\bra}[1]{\langle#1|}
\nc{\braket}[2]{\langle #1 | #2 \rangle}
\nc{\ketbra}[2]{|#1\rangle\!\langle#2|}
\nc{\proj}[1]{| #1\rangle\!\langle #1 |}
\nc{\avg}[1]{\langle#1\rangle}
\nc{\rank}{\operatorname{rank}}
\nc{\tr}{\operatorname{Tr}}
\nc{\ox}{\otimes}
\nc{\dg}{\dagger}
\nc{\dn}{\downarrow}
\nc{\supp}{{\operatorname{supp}}}
\nc{\Span}{{\operatorname{span}}}
\nc{\var}{\operatorname{var}}
\nc{\rar}{\rightarrow}
\nc{\lrar}{\longrightarrow}
\nc{\1}{{\openone}}
\nc{\id}{{\operatorname{id}}}
\nc{\di}{\mathrm{d}}
\nc{\nb}{\nonumber}
\nc{\mc}{\mathcal}

\begin{document}

\title{Operational Interpretation of the Sandwiched R\'enyi Divergence \\ of Order 1/2 to 1 as Strong Converse Exponents}

\author{Ke Li}
    \email{carl.ke.lee@gmail.com}
    \affiliation{Institute for Advanced Study in Mathematics, Harbin Institute of Technology, Harbin
                 150001, China}
\author{Yongsheng Yao}
	\email{yongsh.yao@gmail.com}
    \affiliation{Institute for Advanced Study in Mathematics, Harbin Institute of Technology, Harbin
                 150001, China}	
    \affiliation{School of Mathematics, Harbin Institute of Technology, Harbin 150001, China}
	
\date{\today}

\begin{abstract}
We provide the sandwiched R\'enyi divergence of order $\alpha\in(\frac{1}{2},1)$,
as well as its induced quantum information quantities, with an operational
interpretation in the characterization of the exact strong converse exponents
of quantum tasks. Specifically, we consider (a) smoothing of the max-relative
entropy, (b) quantum privacy amplification, and (c) quantum information decoupling.
We solve the problem of determining the exact strong converse exponents for
these three tasks, with the performance being measured by the fidelity or
purified distance. The results are given in terms of the sandwiched R\'enyi
divergence of order $\alpha\in(\frac{1}{2},1)$, and its induced quantum R\'enyi
conditional entropy and quantum R\'enyi mutual information. This is the first
time to find the precise operational meaning for the sandwiched R\'enyi
divergence with R\'enyi parameter in the interval $\alpha\in(\frac{1}{2},1)$.
\end{abstract}

\maketitle
%\newpage
\tableofcontents
%\newpage

\section{Introduction}
R\'enyi's information divergence, defined for two probability densities, is a fundamental information quantity which has played importance roles in a diversity of fields, ranging from information theory, to probability theory, and to thermodynamics and statistical physics. Its quantum generalization, due to the noncommutativity nature of density matrices, can take infinitely many possible forms. To find which one is the correct quantum generalization is significant and nontrivial.

The sandwiched R\'enyi divergence is one of the proper quantum generalization of R\'enyi's information divergence. For two density matrices $\rho$ and $\sigma$, it is defined as~\cite{MDSFT2013on,WWY2014strong}
\beq
D_{\alpha}^*(\rho\|\sigma):=\frac{1}{\alpha-1} \log \tr\big(\sigma^{\frac{1-\alpha}{2\alpha}}\rho\sigma^{\frac{1-\alpha}{2\alpha}}\big)^\alpha,
\eeq
where $\alpha\in(0,1)\cup(1,\infty)$ is a real parameter. Since its discovery, several operational interpretations for this quantity have been found. On the one hand, with $\alpha\in(1,\infty)$, it characterizes the strong converse exponent for quantum hypothesis testing~\cite{MosonyiOgawa2015quantum, HayashiTomamichel2016correlation}, for classical communication over classical-quantum channels~\cite{MosonyiOgawa2017strong}, for classical data compression with quantum side information~\cite{CHDH2020non}, and for entanglement-assisted or quantum-feedback-assisted communication over quantum channels~\cite{GuptaWilde2015multiplicativity, CMW2016strong, LiYao2022strong}. On the other hand, with $\alpha\in(1,2)$ or $\alpha\in(1,\infty)$, the sandwiched R\'enyi divergence also characterizes the direct error exponent for the smoothing of the max-relative entropy~\cite{LYH2023tight, LiYao2021reliability}, for quantum privacy amplification~\cite{Hayashi2015precise,LYH2023tight}, for quantum information decoupling and state merging~\cite{LiYao2021reliability}, and for quantum channel simulation~\cite{LiYao2021reliable}.

The other proper quantum generalization of R\'enyi's information divergence that has found operational interpretations, is Petz's R{\'e}nyi divergence~\cite{Petz1986quasi}
\beq
D_{\alpha}(\rho\|\sigma):=\frac{1}{\alpha-1} \log \tr\big(\rho^\alpha\sigma^{1-\alpha}\big),
\eeq
for $\alpha\in(0,1)\cup(1,\infty)$. With $\alpha\in(0,1)$, Petz's R{\'e}nyi divergence characterizes the direct error exponent for quantum hypothesis testing in both the symmetric setting~\cite{AKCMBMA2007discriminating,NussbaumSzkola2009chernoff,Li2016discriminating} and the asymmetric setting~\cite{Nagaoka2006converse, Hayashi2007error,ANSV2008asymptotic}, and for classical data compression with quantum side information~\cite{CHDH2020non,Renes2022achievable}. It is also believed to provide the solution to the long-standing open problem of determining the reliability function of classical-quantum channels~\cite{BurnasherHolevo1998on,Winter1999coding,Holevo2000reliability,Hayashi2007error, Dalai2013lower,CHT2019quantum,Renes2022achievable, Cheng2022a}, with $\alpha\in(0,1)$ too.

The roles that the two quantum R{\'e}nyi divergences have played so far cause people to guess that the correct quantum generalization of R{\'e}nyi's information divergence may be
\beq\label{eq:RD-q-f}
\begin{cases}
D_{\alpha}^*(\rho\|\sigma) & \quad\text{when } \alpha\in(1,\infty), \\
D_{\alpha}(\rho\|\sigma)   & \quad\text{when } \alpha\in(0,1).
\end{cases}
\eeq
Indeed, this has been conjectured in the literature; see, e.g.~\cite{MosonyiOgawa2015quantum}.

In this paper, we find an operational interpretation to the sandwiched R\'enyi divergence of order $\alpha\in(\frac{1}{2},1)$. This includes an operational interpretation to $D_{\alpha}^*(\rho\|\sigma)$ itself, as well as to its induced information quantities, the sandwiched R\'enyi conditional entropy and the regularized sandwiched R\'enyi mutual information, all in the interval $\alpha\in(\frac{1}{2},1)$. These results are obtained by considering the strong converse exponents for several quantum information tasks. Employing the purified distance (or, equivalently, the fidelity) as the measure of the performance, we determine these strong converse exponents, which are given in terms of the above-mentioned sandwiched R\'enyi information quantities of order $\alpha\in(\frac{1}{2},1)$. Specifically, our results are as follows.

\begin{itemize}
  \item Smoothing of the max-relative entropy (Section~\ref{subsec:probres_dmax}). The smooth max-relative entropy~\cite{Datta2009min} serves as a basic tool in quantum information sciences. Defined for two states $\rho$ and $\sigma$, it can be regarded as a function of the smoothing quantity $\epsilon\in[0,1]$. Denote its inverse function as $\epsilon(\rho \| \sigma, \lambda)$. The strong converse exponent is the rate of exponential convergence of $\epsilon(\rho^{\ox n} \| \sigma^{\ox n}, nr)$ towards $1$, in the asymptotic limit $n\rar\infty$. We prove that the strong converse exponent is (Theorem~\ref{theorem:dmax})
      \beq\label{eq:intmre}
        \sup_{\frac{1}{2} < \alpha < 1}
        \frac{1-\alpha}{\alpha} \big\{ D^*_{\alpha}(\rho \| \sigma) -r \big\}.
      \eeq
  \item Quantum privacy amplification (Section~\ref{subsec:probres_pa}). Quantum privacy amplification is the art of extracting private randomness from a classical random variable $X$ that is correlated with a quantum system $E$~\cite{BBCM1995generalized,Renner2005security}. It has been a basic primitive in quantum cryptography. Let $\rho_{XE}$ be the joint state of $X$ and $E$ and consider the limiting case that asymptotically many i.i.d.\ copies of the sate $\rho_{XE}$ is available. When the randomness extraction rate $r$ is larger than the conditional entropy $H(X|E)_\rho$, the error goes to $1$ exponentially fast. We show that the best rate of this exponential convergence, namely, the strong converse exponent, is given by
      \beq\label{eq:intpa}
        \sup_{\frac{1}{2} < \alpha < 1} \frac{1-\alpha}{\alpha}\big\{r-H^*_{\alpha}(X|E)_\rho\big\},
      \eeq
      where $H^*_{\alpha}(X|E)_\rho$ is the sandwiched R\'enyi conditional entropy defined in Eq.~\eqref{eq:sRce} (Theorem~\ref{theorem:privacy}).
  \item Quantum information decoupling (Section~\ref{subsec:probres_dec}). Quantum information decoupling is the task of removing the information of a reference system ($R$) from the system under control ($A$), by performing operations on system $A$~\cite{HOW2005partial,HOW2007quantum,ADHW2009mother}. It has been a fundamental quantum information processing task with broad applications~\cite{BDHSW2014quantum,BCR2011the,DARDV2011thermodynamic,BrandaoHorodecki2013area}. Let the joint quantum state of $R$ and $A$ be $\rho_{RA}$. We consider a decoupling strategy via discarding a subsystem of $A$. In the asymptotic limit where many i.i.d.\ copies of $\rho_{RA}$ is available, perfect decoupling can be achieved only when the rate $r$ of qubits discarding is larger than one half of the quantum mutual information $I(R:A)_\rho$. For $r<\frac{1}{2}I(R:A)_\rho$, the error goes to $1$ exponentially fast, and we obtain the best rate of this exponential convergence, namely, the strong converse exponent. It is
      \beq\label{eq:intdec}
        \sup_{\frac{1}{2} < \alpha < 1}
        \frac{1-\alpha}{\alpha}\big\{I_{\alpha}^{*, \rm{reg}}(R:A)_\rho-2r\big\},
      \eeq
      where $I_{\alpha}^{*, \rm{reg}}(R:A)_\rho$ is the regularized version of the sandwiched R\'enyi mutual information, defined in Eq.~\eqref{eq:defrmi} (Theorem~\ref{theorem:maindec}).
\end{itemize}

Because the above results have extended the operational significance of the sandwiched R\'enyi divergence to the range $\alpha\in (\frac{1}{2},1)$, we conclude that the quantum generalization of R\'enyi's information divergence is more complicated than that was conjectured in Eq.~\eqref{eq:RD-q-f}. We point out that, the two intervals $(\frac{1}{2},1)$ and $(1,\infty)$ are related by the H\"older-type duality: if $\frac{1}{2\alpha}+\frac{1}{2\alpha'}=1$ and $\alpha'\in(1,\infty)$, then $\alpha\in(\frac{1}{2},1)$. The sandwiched R\'enyi divergence has nice properties when the order $\alpha$ lies in the range $(\frac{1}{2},1)\cup(1,\infty)$ (cf. Proposition~\ref{prop:mainpro}). However, while the operational interpretation to the sandwiched R\'enyi divergence of order $\alpha\in(1,\infty)$ has been found shortly after the discovery of this fundamental quantity, it had long been open prior to the present paper when $\alpha\in (\frac{1}{2},1)$. Our work has filled this gap.

To derive our results, we employ two different methods. For the problem of smoothing the max-relative entropy, we first prove the special case in which $\rho$ and $\sigma$ commute, using the method of types~\cite{CsiszarKorner2011information}, and then we reduce the general case to the commutative case by exploiting some new properties of the fidelity function. For the problems of quantum privacy amplification and quantum information decoupling, we prove the achievability parts and the optimality parts separately, using different ideas. The achievability parts are technically more involved. We employ an approach developed by Mosonyi and Ogawa~\cite{MosonyiOgawa2017strong}, to derive a weaker bound in terms of the log-Euclidean R\'enyi conditional entropy (mutual information) and then improve it to obtain the final bound. On the other hand, the optimality parts are accomplished by adapting the techniques developed by Leditzky, Wilde and Datta~\cite{LWD2016strong}. Especially, the inequality of Lemma~\ref{lem:LWDg}, which is implicitly proved in~\cite{LWD2016strong} via a delicate application of H\"older's inequality and later formulated in~\cite{WangWilde2019resource}, is crucial for our proof.

\emph{Prior works.}
Leditzky, Wilde and Datta~\cite{LWD2016strong} have proved strong converse theorems with lower bounds to the exponents in terms of the sandwiched R\'enyi entropies, for several quantum information tasks. As for quantum privacy amplification, they have obtained two lower bounds (Theorem 5.3 of~\cite{LWD2016strong}). One is expressed in terms of a difference of R\'enyi entropies. The other one, after reformulation, reads as
\[
  \sup_{\frac{1}{2} < \alpha < 1}
  \frac{1-\alpha}{2\alpha}\big\{r-H^*_{\alpha}(X|E)_\rho\big\}.
      \]
This bound is weaker by a factor of $\frac{1}{2}$ to Eq.~\eqref{eq:intpa}. As regards quantum information decoupling, three lower bounds are given in Theorem 3.8 of~\cite{LWD2016strong} for quantum state splitting, which translate naturally to the problem of quantum information decoupling. The last bound (Eq.~(3.41) in~\cite{LWD2016strong}) is similar to Eq.~\eqref{eq:intdec}, with the quantity $I_{\alpha}^{*, \rm{reg}}(R:A)_\rho$ in Eq.~\eqref{eq:intdec} replaced by $\bar{I}_{\alpha}^*(A:B)_\rho$ (cf.\ Eq.~\eqref{eq:defsrmit}). However, the error/performance function defined in~\cite{LWD2016strong} and that in the present paper are not the same. So they are not comparable. We point out that we are unable to prove the achievability of~\cite[Eq.~(3.41)]{LWD2016strong}, nor can we determine the strong converse exponent under the error/performance function defined in~\cite{LWD2016strong} (see Section~\ref{sec:discussion} for more discussions). Later on, Wang and Wilde~\cite{WangWilde2019resource}, using the method of~\cite{LWD2016strong}, have derived for the strong converse exponent of smoothing the max-relative entropy a lower bound (Eq.~(L14) in~\cite{WangWilde2019resource}), which is exactly the same as Eq.~\eqref{eq:intmre}. So, the optimality part of Theorem~\ref{theorem:dmax} is not new in the present paper. Our proof of Theorem~\ref{theorem:dmax} is not based on the result of~\cite{WangWilde2019resource}, though.

\section{Preliminaries}
  \label{sec:preliminaries}
\subsection{Basic notation}
Let $\mc{H}$ be a Hilbert space associated with a finite-dimensional quantum system. We denote the set of linear operators on $\mc{H}$ as $\mc{L}(\mc{H})$, and the positive semidefinite operators as $\mc{P}(\mc{H})$. We use the notations $\mc{S}(\mc{H})$ and $\mc{S}_{\leq}(\mc{H})$ to represent the set of quantum states and that of subnormalized quantum states, respectively. For $\rho \in \mc{S}(\mc{H})$, $\mc{S}_\rho(\mc{H})$ denotes the set of quantum states whose supports are contained in that of $\rho$, i.e.,
\beq
\mc{S}_\rho(\mc{H}):=\{\tau~|~\tau \in \mc{S}(\mc{H}), ~\supp(\tau) \subseteq \supp(\rho) \},
\eeq
where $\supp(\tau)$ denotes the supporting space of the operator $\tau$. We use $\1_{\mc{H}}$ and $\pi_{\mc{H}}$ to represent the identity operator and the maximally mixed state on $\mc{H}$, respectively. When $\mc{H}$ is associated with a system $A$, the above notations $\mc{L}(\mc{H})$, $\mc{P}(\mc{H})$, $\mc{S}(\mc{H})$, $\mc{S}_{\leq}(\mc{H})$, $\mc{S}_\rho(\mc{H})$, $\1_{\mc{H}}$ and $\pi_{\mc{H}}$ are also written as $\mc{L}(A)$, $\mc{P}(A)$, $\mc{S}(A)$, $\mc{S}_{\leq}(A)$, $\mc{S}_\rho(A)$, $\1_{A}$ and $\pi_A$, respectively. The dimension of system $A$ is denoted by $|A|$. A classical-quantum~(CQ) state is a bipartite state of the form $\rho_{XA}=\sum_{x \in \mc{X}} q(x)\proj{x}_X \ox \rho_A^x$, where $\rho_{A}^x\in\mc{S}(A)$, $\{\ket{x}\}$ is a set of orthonormal basis of $\mc{H}_X$ and $q$ is a probability distribution on $\mc{X}$. We use $\mc{Q}(\mc{X})$ to denote the set of all the probability distributions on $\mc{X}$.

The fidelity between two states $\rho,\sigma\in\mc{S}_{\leq}(\mc{H})$ is given by
\beq
F(\rho,\sigma):=\|\sqrt{\rho}\sqrt{\sigma}\|_1+ \sqrt{(1-\tr\rho)(1-\tr \sigma)},
\eeq
and the purified distance~\cite{GLN2005distance, TCR2009fully} for two states $\rho,\sigma\in\mc{S}_{\leq}(\mc{H})$ is defined as
\beq\label{eq:defP}
P(\rho, \sigma):=\sqrt{1-F^2(\rho, \sigma)}.
\eeq

A quantum channel or quantum operation $\mc{N}_{A \rar B}$ is a completely positive and trace-preserving~(CPTP) linear map from $\mc{L}(A)$ to $\mc{L}(B)$. The Stinespring dilation theorem~\cite{Stinespring1955positive} states that there exist an ancillary system $C$ in a pure state $\proj{0}_C$, a system $E$ and a unitary transformation $U_{AC \rar BE}$ such that
$\mc{N}_{A \rar B} (X)= \tr_{E} \big[U \big(X \ox \proj{0}_C\big) U^\dg \big]$, for $X\in\mc{L}(A)$. A quantum measurement is characterized by a set of positive semidefinite operators $\{M_x\}_x$ which satisfies $\sum_xM_x=\1$. When we make this measurement on a state $\rho$, we get the result $x$
with probability $\tr\rho M_x$.

Let $H$ be a self-adjoint operator on $\mc{H}$ with spectral projections $P_1, \ldots, P_r$. Then the pinching map associated with $H$ is defined as
\beq
\mc{E}_{H}:X\mapsto\sum_{i=1}^r P_i X P_i.
\eeq
The pinching inequality~\cite{Hayashi2002optimal} tells that for any $\sigma \in \mc{P}(\mc{H})$,
we have
\begin{equation}
\sigma \leq v(H) \mc{E}_{H}(\sigma),
\end{equation}
where $v(H)$ is the number of different eigenvalues of $H$.

Let $S_n$ be the symmetric group of the permutations of $n$ elements. The set of symmetric states on $A^n$ is defined as
\beq
\mc{S}_{\rm{sym}}(A^n):=\left\{\sigma_{A^n} \in \mc{S}(A^n)~|~
W^{\iota}_{A^n} \sigma_{A^n}W^{\iota \dg }_{A^n}=\sigma_{A^n}, \ \forall\ \iota \in S_n\right\}.
\eeq
where $W^{\iota}_{A^n}: \ket{\psi_1} \ox \ldots \ox \ket{\psi_n}
\mapsto  \ket{\psi_{\iota^{-1}(1)}} \ox \ldots \ox \ket{\psi_{\iota^{-1}(n)}}$ is the natural representation of $\iota\in S_n$. There exists a single symmetric state that dominates all the other symmetric states, in the sense of the following Lemma~\ref{lem:sym}. Two different constructions of this symmetric state can be found in~\cite{Hayashi2009universal} and~\cite{CKR2009postselection}, respectively. See~\cite[Lemma 1]{HayashiTomamichel2016correlation} and~\cite[Appendix A]{MosonyiOgawa2017strong} for detailed arguments.
\begin{lemma}
\label{lem:sym}
For a quantum system $A$ of finite dimension and any $n\in\mathbb{N}$, there exists a single symmetric state $\sigma_{A^n}^u \in \mc{S}_{\rm{sym}}(A^n)$ such that every symmetric state $\sigma_{A^n} \in \mc{S}_{\rm{sym}}(A^n)$ is dominated as
\begin{equation}
\sigma_{A^n} \leq v_{n,|A|}\sigma_{A^n}^u,
\end{equation}
where $v_{n,|A|} \leq (n+1)^{\frac{(|A|+2)(|A|-1)}{2}}$ is a polynomial of $n$. The number of distinct eigenvalues of $\sigma_{A^n}^u$ is upper bounded by $v_{n,|A|}$ as well.
\end{lemma}

\subsection{Quantum R\'enyi divergences and entropies}
For $\rho \in \mc{S}(\mc{H})$ and $\sigma \in \mc{P}(\mc{H})$, the quantum relative
entropy of $\rho$ and $\sigma$ is defined as~\cite{Umegaki1954conditional}
\beq
D(\rho\|\sigma):= \begin{cases}
\tr(\rho(\log\rho-\log\sigma)) & \text{ if }\supp(\rho)\subseteq\supp(\sigma), \\
+\infty                        & \text{ otherwise.}
                  \end{cases}
\eeq
Let $\rho_{AB}$ be a bipartite quantum state, the conditional entropy and mutual information
of $\rho_{AB}$ are defined, respectively as
\begin{align}
H(A|B)_\rho:&= -\min_{\sigma_B \in \mc{S}(B)} D(\rho_{AB} \| \1_A \ox \sigma_B), \\
I(A:B)_\rho:&=\min_{\sigma_B \in \mc{S}(B)} D(\rho_{AB} \| \rho_A \ox \sigma_B)
=\min_{\sigma_A \in \mc{S}(A), \atop \sigma_B \in \mc{S}(B)} D(\rho_{AB} \| \sigma_A \ox \sigma_B).
\end{align}

The classical R\'enyi divergence is an important information quantity which has lots of applications~\cite{Renyi1961measures}. When it comes to the quantum setting, the situation
appears to be more complicated, because there can be infinitely many inequivalent quantum
generalizations of the R{\'e}nyi divergence. In this paper, we are interested in the
\emph{sandwiched R{\'e}nyi divergence}~\cite{MDSFT2013on, WWY2014strong} and the
\emph{log-Euclidean R{\'e}nyi divergence}~\cite{OgawaNagaoka2000strong, MosonyiOgawa2017strong}.
\begin{definition}
\label{definition:sand}
Let $\alpha\in(0,+\infty)\setminus\{1\}$, $\rho\in\mc{S}(\mc{H})$ and $\sigma\in\mc{P}(\mc{H})$.
When $\alpha >1$ and $\supp(\rho)\subseteq\supp(\sigma)$ or $\alpha\in (0,1)$ and $\supp(\rho)\not\perp\supp(\sigma)$, the sandwiched R{\'e}nyi divergence of order $\alpha $
is defined as
\beq
D_{\alpha}^*(\rho \| \sigma):=\frac{1}{\alpha-1} \log Q_{\alpha}^*(\rho \| \sigma),
\quad\text{with}\ \
Q_{\alpha}^*(\rho \| \sigma)=\tr {({\sigma}^{\frac{1-\alpha}{2\alpha}} \rho {\sigma}^{\frac{1-\alpha}{2\alpha}})}^\alpha;
\eeq
otherwise, we set $D_{\alpha}^*(\rho \| \sigma)=+\infty$. When $\alpha > 1$ and $\supp(\rho)\subseteq\supp(\sigma)$  or  $\alpha \in (0,1)$ and $\supp(\rho)\cap\supp(\sigma)\neq\{0\}$, the log-Euclidean R{\'e}nyi divergence of order
$\alpha $ is defined as
\beq
D_{\alpha}^{\flat}(\rho \| \sigma):=\frac{1}{\alpha-1}\log Q_{\alpha}^{\flat}(\rho \| \sigma),
\quad\text{with}\ \
Q_{\alpha}^{\flat}(\rho \| \sigma)=\tr 2^{\alpha \log \rho +(1-\alpha) \log \sigma};
\eeq
otherwise, we set $D_{\alpha}^{\flat}(\rho \| \sigma)=+\infty$. Here if $\rho$ or $\sigma$ is not of full rank, the expression for $Q_{\alpha}^{\flat}$ should be understood as $\lim\limits_{\epsilon\searrow 0}\tr 2^{\alpha\log(\rho+\epsilon\1)+(1-\alpha)\log(\sigma+\epsilon\1)}$.
\end{definition}

For some range of the R{\'e}nyi parameter $\alpha$, the log-Euclidean R{\'e}nyi divergence
has appeared already in~\cite{HiaiPetz1993the} and~\cite{OgawaNagaoka2000strong}.
However, its importance was not recognized until recently. In~\cite{MosonyiOgawa2017strong},
Mosonyi and Ogawa have employed it as an intermediate quantity to derive the strong
converse exponent for classical-quantum channels. Its name was given in~\cite{CHT2019quantum}.

In the way as the definition of the conditional entropy, we define the sandwiched
conditional R\'enyi entropy~\cite{MDSFT2013on} and the log-Euclidean conditional R\'enyi entropy~\cite{CHDH2020non} of order $\alpha \in (0,+\infty)\setminus\{1\}$ for a
bipartite state $\rho_{AB} \in \mc{S}(AB)$ as
\begin{align}\label{eq:sRce}
H_{\alpha}^{*}(A|B)_\rho &:=-\min_{\sigma_B\in\mc{S}(B)}D_{\alpha}^*(\rho_{AB}\|\1_A\ox\sigma_B), \\
H_{\alpha}^{\flat}(A|B)_\rho &:=-\min_{\sigma_B\in\mc{S}(B)}D_{\alpha}^{\flat}(\rho_{AB}\|\1_A\ox\sigma_B).
\end{align}
Similarly, the sandwiched R\'enyi mutual information~\cite{WWY2014strong, Beigi2013sandwiched, MckinlayTomamichel2020decomposition} and the log-Euclidean R\'enyi mutual information of
order $\alpha \in (0,+\infty)\setminus\{1\}$ for $\rho_{AB} \in \mc{S}(AB)$ are defined as
\begin{align}
I_{\alpha}^*(A:B)_\rho   &:=\min_{\sigma_A\in\mc{S}(A),\atop\sigma_B\in\mc{S}(B)}D_{\alpha}^*(\rho_{AB}\|\sigma_A\ox\sigma_B),
   \label{eq:srmi-de} \\
I_{\alpha}^{\flat}(A:B)_\rho
&:=\min_{\sigma_A\in\mc{S}(A),\atop\sigma_B\in\mc{S}(B)}D_{\alpha}^{\flat}(\rho_{AB}\|\sigma_A\ox\sigma_B),\\
\bar{I}_{\alpha}^*(A:B)_\rho
&:=\min_{\sigma_B\in\mc{S}(B)} D_{\alpha}^*(\rho_{AB} \| \rho_A \ox \sigma_B), \label{eq:defsrmit} \\
\bar{I}_{\alpha}^{\flat}(A:B)_\rho
&:=\min_{\sigma_B\in\mc{S}(B)} D_{\alpha}^{\flat}(\rho_{AB} \| \rho_A \ox \sigma_B).
\end{align}
We also define the regularized version of the sandwiched R\'enyi mutual information
for $\rho_{AB}$ as
\beq\label{eq:defrmi}
I_{\alpha}^{*, \rm{reg}}(A:B)_\rho
:= \lim_{n\rar\infty} \frac{1}{n}I_{\alpha}^*(A^n:B^n)_{\rho^{\ox n}}.
\eeq

\begin{remark}
Different from the quantum mutual information, the two types of the R\'enyi mutual
information defined above are not equal to each other in general. Namely, usually we
have $I_{\alpha}^*(A:B)_\rho\neq \bar{I}_{\alpha}^*(A:B)_\rho$ and
$I_{\alpha}^{\flat}(A:B)_\rho\neq \bar{I}_{\alpha}^{\flat}(A:B)_\rho$.
\end{remark}

The above definitions can be extended to the cases $\alpha=0$, $\alpha=1$ and
$\alpha=+\infty$, by taking the limits. In the following proposition, we collect
some properties of the R\'enyi information quantities.

\begin{proposition}
\label{prop:mainpro}
Let $\rho \in \mc{S}(\mc{H})$ and $\sigma \in \mc{P}(\mc{H})$. The sandwiched R{\'e}nyi
divergence and the log-Euclidean R{\'e}nyi divergence satisfy the following properties.
\begin{enumerate}[(i)]
  \item Monotonicity in R{\'e}nyi parameter~\cite{MDSFT2013on,Beigi2013sandwiched,
      MosonyiOgawa2017strong}: if $0\leq \alpha \leq \beta$, then
      $D_{\alpha}^{(t)}(\rho \| \sigma) \leq  D^{(t)}_{\beta}(\rho \| \sigma)$,
      for $(t)=*$ and $(t)=\flat$;
  \item Monotonicity in $\sigma$~\cite{MDSFT2013on,MosonyiOgawa2017strong}: if $\sigma' \geq \sigma$, then $D_{\alpha}^{(t)}(\rho \| \sigma') \leq D_{\alpha}^{(t)}(\rho \| \sigma)$,
      for $(t)=*$, $\alpha \in [\frac{1}{2},+\infty)$ and for $(t)=\flat$, $\alpha \in [0,+\infty)$;
  \item Variational representation~\cite{MosonyiOgawa2017strong}: the log-Euclidean
      R{\'e}nyi divergence has the following variational representation
      \beq
      D_{\alpha}^{\flat}(\rho \| \sigma)= \begin{cases}
         \min\limits_{\tau \in \mc{S}(\mc{H})} \big\{D(\tau \| \sigma)
         -\frac{\alpha}{\alpha-1}D(\tau \| \rho)\big\}, & \alpha \in (0,1), \\
         \max\limits_{\tau \in \mc{S}_\rho(\mc{H})} \big\{D(\tau \| \sigma)
         -\frac{\alpha}{\alpha-1}D(\tau \| \rho)\big\}, & \alpha \in (1,+\infty);
      \end{cases}
      \eeq
  \item Data processing inequality~\cite{FrankLieb2013monotonicity, Beigi2013sandwiched, MDSFT2013on, WWY2014strong, MosonyiOgawa2017strong}: letting $\mc{N}$ be a CPTP map from $\mc{L}(\mc{H})$ to $\mc{L}(\mc{H}')$, we have
      \beq
      D_{\alpha}^{(\rm{t})}(\mc{N}(\rho) \| \mc{N}(\sigma)) \leq D_{\alpha}^{(\rm{t})}(\rho \| \sigma),
      \eeq
      for $(t)=\flat$, $\alpha \in [0,1]$ and for $(t)=*$, $\alpha \in [\frac{1}{2},+\infty)$;
  \item Monotonicity under discarding classical information~\cite{LWD2016strong}: for the state $\rho_{XAB}$ that is classical on $X$ and for $\alpha\in(0,+\infty)$, we have
      \begin{equation}
      H_{\alpha}^*(XA|B)_\rho \geq H_{\alpha}^*(A|B)_\rho;
      \end{equation}
  \item Dimension bound for the sandwiched R\'enyi mutual information~\cite{LWD2016strong}: for a tripartite state $\rho_{ABC} \in \mc{S}(ABC)$ and for $\alpha \in [\frac{1}{2}, +\infty)$, we have
      \begin{equation}
        \bar{I}_{\alpha}^{*}(A:BC)_\rho \leq \bar{I}_{\alpha}^{*}(A:B)_\rho + 2\log|C|.
      \end{equation}
\end{enumerate}
\end{proposition}

\subsection{Method of types}
The method of types is a useful tool in information theory. Here we recall the relevant definitions and properties briefly. For an overall introduction, we refer the interested readers
to Reference~\cite{CsiszarKorner2011information}.

Let $\mc{X}$ be a finite alphabet set and $n \in \mathbb{N}$. For a sequence $x^n \in \mc{X}^{\times n}$, the type $t_{x^n}$ is the empirical distribution of $x^n$, i.e.,
\beq
t_{x^n}(a)=\sum_{i=1}^n \frac{\delta_{x_i,a}}{n}, \quad \forall a \in \mc{X}.
\eeq
The set of all types is denoted as $\mc{T}_n^{\mc{X}}$. The size of $\mc{T}_n^{\mc{X}}$ is bounded by
\begin{equation}
\label{eq:typenumber}
|\mc{T}_n^{\mc{X}}| \leq (n+1)^{|\mc{X}|}.
\end{equation}
If $t\in\mc{T}_n^{\mc{X}}$, then the set of all the sequences of type $t$ is called the type class of $t$, denoted as $T_n^t$, that is,
\beq
T_n^t:=\{x^n~|~t_{x^n}=t\}.
\eeq
The number of the sequences in a type class $T_n^t$ satisfies the following relation:
\begin{equation}
\label{eq:numt}
(n+1)^{-|\mc{X}|} 2^{nH(t)} \leq |T_n^t| \leq 2^{nH(t)},
\end{equation}
where $H(t)=-\sum_{x\in\mc{X}} t(x)\log t(x)$ is the Shannon entropy of $t$. Let $X_1,X_2,\ldots,X_n$ be a sequence of i.i.d.\ random variables each taking values in $\mc{X}$ according to the distribution $p$. Then the probability of $X^n$ being equal to $x^n$ of type $t$ is given by
\begin{equation}
\label{eq:prot}
p^n(x^n)=\prod_{i=1}^np(x_i)=2^{-nH(t)-nD(t\|p)},
\end{equation}
where $D(t\|p)=\sum_{x\in\mc{X}}t(x)\log\frac{t(x)}{p(x)}$ is the relative entropy. The combination of Eq.~(\ref{eq:numt}) and Eq.~(\ref{eq:prot}) gives that the probability of $X^n$ taking values in a type class $T_n^t$, given by $\sum_{x^n\in T_n^t}p^n(x^n)$, can be bounded as
\begin{equation}
\label{eq:proset}
(n+1)^{-|\mc{X}|}2^{-nD(t\|p)} \leq \sum_{x^n\in T_n^t}p^n(x^n) \leq 2^{-nD(t\|p)}.
\end{equation}

\section{Problems and Main Results}
  \label{sec:problems-results}
\subsection{Smoothing of the max-relative entropy}
  \label{subsec:probres_dmax}
Let $\rho \in \mc{S}(\mc{H})$, $\sigma \in \mc{P}(\mc{H})$, and $\epsilon\in[0,1]$. The smooth max-relative entropy based on the purified distance is defined as~\cite{Datta2009max}
\beq\label{eq:dmax-def}
D^{\epsilon}_{\rm{max}}(\rho \| \sigma)
:= \inf\big\{\lambda\in\mathbb{R} ~\big|~ \big(\exists\tilde{\rho}\in\mc{B}^{\epsilon}(\rho)\big)~ \tilde{\rho}\leq 2^{\lambda} \sigma \big\},
\eeq
where $\mc{B}^{\epsilon}(\rho):=\{\tilde{\rho}\in\mc{S}_{\leq}(\mc{H})~|~P(\tilde{\rho},\rho) \leq \epsilon\}$ is the $\epsilon$-ball of (subnormalized) quantum states around $\rho$. Regarding
$D^{\epsilon}_{\rm{max}}(\rho \| \sigma)$ as a function of $\epsilon$, we introduce its inverse function
\begin{align}
\epsilon(\rho \| \sigma, \lambda)
:=&\min\left\{\epsilon \;|\; D^{\epsilon}_{\rm{max}}(\rho\|\sigma)\leq\lambda\right\}\nonumber \\
 =&\min \big\{P(\rho, \tilde{\rho})~\big|~\tilde{\rho} \in \mc{S}_{\leq}(\mc{H}) \text{ and } \tilde{\rho} \leq 2^\lambda \sigma \big\} \label{eq:invdmax-def}
\end{align}
and call it the smoothing quantity.

The smooth max-relative entropy is not only a basic tool in quantum information theory, but it also quantifies in an exact way an operational task. In~\cite{WangWilde2019resource}, it is shown that $D^{\epsilon}_{\rm{max}}(\rho \| \sigma)$ is the cost of preparing the box $(\rho, \sigma)$ with error $\epsilon$, in the resource theory of asymmetric distinguishability. So, the study of $D^{\epsilon}_{\rm{max}}(\rho\|\sigma)$ is equivalent to the study of the dilution problem of asymmetric distinguishability.

The quantum asymptotic equipartition property~\cite{TCR2009fully, Tomamichel2015quantum} states that
\beq\label{eq:AEP}
\lim_{n\rar\infty}\frac{1}{n}D^{\epsilon}_{\rm{max}}(\rho^{\ox n}\|\sigma^{\ox n})=D(\rho\|\sigma).
\eeq
Eq.~\eqref{eq:AEP} is further strengthened by the second-order asymptotics established in~\cite{TomamichelHayashi2013hierarchy}. The large-deviation type behaviors are more conveniently stated in terms of $\epsilon(\rho\|\sigma, \lambda)$. When $r>D(\rho\|\sigma)$, the rate of exponential decay of $\epsilon(\rho^{\ox n}\|\sigma^{\ox n}, nr)$ towards $0$ is determined recently in~\cite{LYH2023tight}. On the other hand, when $r<D(\rho\|\sigma)$, the smoothing quantity  $\epsilon(\rho^{\ox n}\|\sigma^{\ox n}, nr)$ must converge to $1$ exponentially fast; see~\cite{WangWilde2019resource,SalzmannDatta2022total,
Wilde2022distinguishability} for bounds using R\'enyi relative entropies. The exact rate of this exponential convergence is called the strong converse exponent, which was unknown.

In this paper, we determine the strong converse exponent for the smoothing of the max-relative entropy. Our result is the following Theorem~\ref{theorem:dmax}. We point out that the optimality (``$\geq$'') part has been derived in~\cite{WangWilde2019resource} (see Eq.~(L14) of~\cite[Appendix L]{WangWilde2019resource}). Our proof, dealing with the achievability part and the optimality part in a unified way, does not depend on the result of~\cite{WangWilde2019resource}, however.
\begin{theorem}
\label{theorem:dmax}
For $\rho \in \mc{S}(\mc{H})$, $\sigma \in \mc{P}(\mc{H})$ and $r \in \mathbb{R}$, we have
\begin{equation}\label{eq:dmax}
\lim_{n \rar \infty} \frac{-1}{n} \log \big(1-\epsilon(\rho^{\ox n} \| \sigma^{\ox n}, nr) \big)
=\sup_{\frac{1}{2}\leq \alpha < 1}\frac{1-\alpha}{\alpha}\big\{D^*_{\alpha}(\rho \| \sigma)-r\big\}.
\end{equation}
\end{theorem}

In Eq.~\eqref{eq:dmax}, the range $\frac{1}{2}\leq\alpha<1$ can be replaced by $\frac{1}{2}<\alpha<1$, in agreement with Eq.~\eqref{eq:intmre}. This is because the expression under optimization is either right continuous at $\frac{1}{2}$, or being $+\infty$ on the whole interval. Some more remarks are as follows. (\romannumeral1) When $\supp(\rho)\subseteq\supp(\sigma)$, the expression under optimization is bounded and continuous as a function of $\alpha$ on $[\frac{1}{2},1]$. So the optimization can be replaced by taking the maximum over $\alpha\in[\frac{1}{2},1]$.  (\romannumeral2) When $\supp(\rho)\nsubseteq\supp(\sigma)$ and $\supp(\rho)\not\perp\supp(\sigma)$, the expression under optimization is indefinite at $\alpha=1$, because $D^*_{\alpha}(\rho \| \sigma)|_{\alpha=1}=D(\rho \| \sigma)=+\infty$ but $\frac{1-\alpha}{\alpha}|_{\alpha=1}=0$. However, we can check that, for $\frac{1}{2}\leq\alpha<1$,
\begin{equation}\label{eq:dmax-alt}
\frac{1-\alpha}{\alpha}\big\{D^*_{\alpha}(\rho \| \sigma)-r\big\}
=\frac{1-\alpha}{\alpha}\big\{D^*_{\alpha}(\rho' \| \sigma)-r\big\}-\log\tr[\Pi_\sigma\rho\Pi_\sigma],
\end{equation}
where $\Pi_\sigma$ is the projection onto the supporting space of $\sigma$ and $\rho'$ is the normalized version of $\Pi_\sigma\rho\Pi_\sigma$. So, using the alternative expression, the optimization can again be replaced by taking the maximum over $\alpha\in[\frac{1}{2},1]$. (\romannumeral3) When $\supp(\rho)\perp\supp(\sigma)$, Eq.~\eqref{eq:dmax} is $+\infty$.

\subsection{Quantum privacy amplification}
  \label{subsec:probres_pa}
Consider a CQ state
\begin{equation}
  \label{eq:source-state}
\rho_{XE}=\sum_{x\in\mc{X}} p_x \proj{x}_X \ox \rho^x_E.
\end{equation}
Let the system $X$, which is also regarded as a classical random variable, represent an imperfect random number that is partially correlated with an adversary Eve's system $E$. In the procedure of privacy amplification, we apply a hash function $f:\mc{X}\rightarrow\mc{Z}$ on $X$ to extract a random number $Z$, which is expected to be uniformly distributed and independent of the adversary's system $E$. The action of the hash function $f$ can be written as a quantum operation
\begin{equation}
  \label{eq:hashq}
\mc{P}_f: \omega\mapsto\sum_{x\in\mc{X}}\bra{x}\omega\ket{x}\proj{f(x)}.
\end{equation}
So the resulting state of privacy amplification is
\begin{equation}
  \label{eq:final-state}
\mc{P}_f(\rho_{XE})=\sum_{z\in\mc{Z}} \proj{z}_Z \ox \sum_{x\in f^{-1}(z)} p_x\rho^x_E.
\end{equation}
The effect is measured by two quantities. One is the size of the extracted randomness, $\log|\mc{Z}|$ in bits. The other one is the security parameter, defined as
\beq\label{eq:def-p-pa}
\mathfrak{P}^{\rm pa}(\rho_{XE},f)
:=\max_{\omega_E\in\mc{S}(E)}F^2\big(\mc{P}_f(\rho_{XE}),\pi_Z\ox\omega_E \big),
\eeq
where $\pi_Z$ is the maximally mixed state. Since the purified distance is a function of fidelity (cf.\ Eq.~\eqref{eq:defP}), the security parameter employed here takes the same information as the one based on purified distance in previous works (e.g.,~\cite{TomamichelHayashi2013hierarchy, ABJT2020partially, LYH2023tight}). See more discussions in Remark~\ref{rk:fp}.

In the asymptotic setting where an arbitrary large number of copies of the state $\rho_{XE}$ is available, we apply the hash function $f_n:\mc{X}^{\times n} \rar \mc{Z}_n$ to extract private randomness from $\rho_{XE}^{\ox n}$, for any $n \in \mathbb{N}$. It has been proven in~\cite{DevetakWinter2005distillation} that to achieve asymptotically perfect privacy amplification such that $\mathfrak{P}^{\rm pa}(\rho_{XE}^{\ox n},f_n)\rar 1$, the rate of randomness extraction must satisfy
\beq
\limsup_{n\rar\infty} \frac{1}{n} \log |\mc{Z}_n| \leq H(X|E)_\rho.
\eeq
Finer asymptotic results, including the second-order expansion based on purified distance~\cite{TomamichelHayashi2013hierarchy} and that based on trace distance~\cite{SGC2022optimal}, as well as the large-deviation type error exponent~\cite{LYH2023tight}, have been obtained later.

On the other hand, when the rate of randomness extraction is larger than $H(X|E)_\rho$, the strong converse property holds. Specifically, for any sequence of hash functions $\{f_n:
\mc{X}^{\times n} \rar \mc{Z}_n\}_{n \in \mathbb{N}}$, we have
\beq\label{eq:scp-pa}
\liminf_{n\rar\infty}\frac{1}{n} \log |\mc{Z}_n| > H(X|E)_\rho  \quad\Rightarrow\quad \lim_{n\rar\infty}\mathfrak{P}^{\rm pa}(\rho_{XE}^{\ox n},f_n) = 0,
\eeq
and the decay of $\mathfrak{P}^{\rm pa}(\rho_{XE}^{\ox n},f_n)$ is exponentially fast. This can be seen, from the one-shot converse bound in terms of the smooth conditional min-entropy~\cite{TSSR2011leftover, TomamichelHayashi2013hierarchy} combined with the asymptotic equipartition property~\cite{TCR2009fully}. The work~\cite{LWD2016strong} proved the strong converse property by providing a bound on the rate of exponential decay in terms of the sandwiched R\'enyi conditional entropy; see also recent works~\cite{SGC2022strong,SalzmannDatta2022total} for bounds employing Petz's R\'enyi conditional entropy. The optimal achievable exponent of this decay is called the strong converse exponent and is defined as
\begin{equation}\label{eq:defscpa}
E_{\rm{sc}}^{\rm{pa}}(\rho_{XE}, r):=\inf \left\{\limsup_{n\rar\infty} \frac{-1}{n} \log \mathfrak{P}^{\rm pa}(\rho_{XE}^{\ox n},f_n)~\Big|~ \liminf_{n\rar\infty} \frac{1}{n} \log |\mc{Z}_n| \geq r \right\}.
\end{equation}

\begin{remark}\label{rk:fp}
The above definition based on the squared fidelity (cf.\ Eq.~\eqref{eq:def-p-pa}) is equivalent to a definition of the strong converse exponent based on the purified distance. Specifically, in Eq.~\eqref{eq:defscpa}, $\mathfrak{P}^{\rm pa}(\rho_{XE}^{\ox n},f_n)$ can be replaced by $1-\widehat{\mathfrak{P}}^{\rm pa}(\rho_{XE}^{\ox n},f_n)$ without changing the strong converse exponent, where
\beq
\widehat{\mathfrak{P}}^{\rm pa}(\rho_{XE},f)
:=\min\limits_{\omega_E\in\mc{S}(E)}P\big(\mc{P}_f(\rho_{XE}),\pi_Z\ox\omega_E \big)
\eeq
is a security parameter widely used in the literature. This fact can be easily verified by the relation $P=\sqrt{1-F^2}$.
\end{remark}

In this paper, we derive the exact expression for $E_{\rm{sc}}^{\rm{pa}}(\rho_{XE}, r)$. The result is given in the following theorem.
\begin{theorem}
\label{theorem:privacy}
Let $\rho_{XE}$ be a CQ state. For any rate $r\geq 0$, we have
\begin{equation}
E_{\rm{sc}}^{\rm{pa}}(\rho_{XE}, r)=\sup_{\frac{1}{2} \leq \alpha \leq 1} \frac{1-\alpha}{\alpha}\big\{r-H^*_{\alpha}(X|E)_\rho\big\}.
\end{equation}
\end{theorem}

In Theorem~\ref{theorem:privacy}, the expression under optimization is continuous on the closed interval $[\frac{1}{2},1]$. So, the supremum can be replaced by a maximum. For the same reason, the range $\frac{1}{2} \leq \alpha \leq 1$ can be replaced by $\frac{1}{2} < \alpha < 1$, in agreement with Eq.~\eqref{eq:intpa}.

\subsection{Quantum information decoupling}
  \label{subsec:probres_dec}
Let $\rho_{RA}\in\mc{S}(RA)$ be a bipartite quantum state with $A$ in the lab and $R$ held by a referee. Quantum information decoupling is the task of removing the correlation between system $A$ and system $R$, by performing quantum operations on $A$. We focus on decoupling strategy via discarding a subsystem~\cite{ADHW2009mother}; other strategies, such as those of~\cite{HOW2007quantum} and~\cite{GPW2005quantum}, can be treated similarly. A general decoupling scheme $\mc{D}$ consists of a catalytic system $A'$ in a state $\sigma_{A'}$ and a unitary transformation $U : \mc{H}_{AA'} \rar \mc{H}_{\bar{A}\tilde{A}}$. We write
\beq\label{eq:decscheme}
\mc{D}:=(\sigma_{A'},\ U: \mc{H}_{AA'} \rar \mc{H}_{\bar{A}\tilde{A}}).
\eeq
Discarding the subsystem $\tilde{A}$, the goal of quantum information decoupling is to make the resulting
state on $R$ and $\bar{A}$ close to a product form. Thus the performance of this scheme is characterized by \beq\label{eq:def-p-dec}
\mathfrak{P}^{\rm dec}(\rho_{RA},\mc{D})
:=\max_{\omega_{R}\in\mc{S}(R),\atop\omega_{\bar{A}}\in\mc{S}(\bar{A})}
F^2\big(\tr_{\tilde{A}}[ U (\rho_{RA} \ox \sigma_{A'}) U^\dg ],\omega_R \ox \omega_{\bar{A}} \big).
\eeq
The cost is measured by the amount of discarded qubits, namely, $\log |\tilde{A}|$. When the catalytic state $\sigma_{A'}$ in Eq.~\eqref{eq:decscheme} is restricted to be pure, this is the standard setting of decoupling considered originally~\cite{HOW2005partial,HOW2007quantum,ADHW2009mother}. Equivalently, in standard decoupling the catalytic system is of dimension one and can be omitted, and the unitary transformation is replaced by an isometry $U : \mc{H}_{A} \rar \mc{H}_{\bar{A}\tilde{A}}$. The definition that we adopt here, was introduced in~\cite{MBDRC2017catalytic,ADJ2017quantum} and called catalytic decoupling.

It has been established that, when arbitrarily many copies of the state $\rho_{RA}$ is available, asymptotically perfect decoupling can be achieved if and only if the rate of decoupling cost is at least $\frac{1}{2} I(R:A)_\rho$~\cite{ADHW2009mother}. This result holds for both the standard and the catalytic settings~\cite{MBDRC2017catalytic,ADJ2017quantum}. With catalyst, the second-order asymptotics has been derived in~\cite{MBDRC2017catalytic}. Recently, in the catalytic setting too, we have conducted the exponential analysis, obtaining the best exponent for the convergence of the performance towards the perfect in case that the rate of decoupling cost is below a critical value~\cite{LiYao2021reliability}.

When the rate of decoupling cost is smaller than $\frac{1}{2} I(R:A)_\rho$, the strong converse property states that for any sequence of decoupling schemes $\big\{\mc{D}_n=(\sigma_{A'_n},\ U_n:\mc{H}_{A^nA'_n} \rar \mc{H}_{\bar{A}_n\tilde{A}_n}) \big\}_{n \in \mathbb{N}}$, we have
\beq
\limsup_{n\rar\infty}\frac{1}{n} \log|\tilde{A}_n| < \frac{1}{2} I(R:A)_\rho  \quad\Rightarrow\quad \lim_{n\rar\infty}\mathfrak{P}^{\rm dec}(\rho_{RA}^{\ox n},\mc{D}_n) = 0.
\eeq
Moreover, the convergence is exponentially fast. This follows from the one-shot smooth-entropy bound~\cite{BCR2011the} coupled with the asymptotic equipartition property~\cite{TCR2009fully, Tomamichel2015quantum}; see also~\cite{Sharma2014a,LWD2016strong} for proofs employing R\'enyi entropies, and~\cite{MBDRC2017catalytic} for discussions on the catalytic setting. The optimal rate of exponential decay of $\mathfrak{P}^{\rm dec}(\rho_{RA}^{\ox n},\mc{D}_n)$ in the strong converse domain, for a fixed rate of decoupling cost, is called the strong converse exponent. Formally, it is defined as
\begin{equation} \label{eq:defsce}
E_{\rm{sc}}^{\rm{dec}}(\rho_{RA},r):=
\inf \left\{\limsup_{n\rar\infty} \frac{-1}{n}\log \mathfrak{P}^{\rm dec}(\rho_{RA}^{\ox n},\mc{D}_n)~\Big|~\limsup_{n\rar\infty}\frac{1}{n} \log|\tilde{A}_n| \leq r\right\}.
\end{equation}
For the same reason as Remark~\ref{rk:fp}, our definition here based on the squared fidelity (cf.\ Eq.~\eqref{eq:def-p-dec}) is equivalent to a definition based on purified distance, making similar modifications.

\smallskip
In this paper, we derive the exact expression for $E_{\rm{sc}}^{\rm{dec}}(\rho_{RA},r)$. Our result is as follows.
\begin{theorem}
\label{theorem:maindec}
Let $\rho_{RA} \in \mc{S}(RA)$ and $r \geq 0$. We have
\begin{equation} \label{eq:maindec}
E_{\rm{sc}}^{\rm{dec}}(\rho_{RA},r)=\sup_{\frac{1}{2} \leq \alpha \leq 1} \frac{1-\alpha}{\alpha}\big\{I_{\alpha}^{*, \rm{reg}}(R:A)_\rho-2r\big\}.
\end{equation}
\end{theorem}

Although Theorem~\ref{theorem:maindec} is stated for the catalytic setting, it applies for standard decoupling as well. In the proof, we derive the achievability part by employing a sequence of standard decoupling schemes, while we show the optimality part for any catalytic decoupling schemes. Thus, we conclude that the right hand side of Eq.~\eqref{eq:maindec} is actually the strong converse exponent for both settings.

The regularized formula in Eq.~\eqref{eq:maindec} makes it difficult to understand the properties. If $I_{\alpha}^*$ is additive for any product state $\rho_{R_1A_1}\ox\sigma_{R_2A_2}$, in the sense that $I_{\alpha}^*(R_1R_2:A_1A_2)_{\rho\ox\sigma}=I_{\alpha}^*(R_1:A_1)_\rho+I_{\alpha}^*(R_2:A_2)_\sigma$, then $I_{\alpha}^{*, \rm{reg}}=I_{\alpha}^*$ and the regularization is not needed. This additivity was shown in the classical case for $\alpha\in[\frac{1}{2},\infty)$~\cite{TomamichelHayashi2017operational}, and recently proved in the quantum case for $\alpha\in(1,\infty)$~\cite{ChengGao2023tight}. We conjecture that it holds as well for $\alpha\in[\frac{1}{2},1)$. However, we are unable to prove it, and the methods of~\cite{TomamichelHayashi2017operational,ChengGao2023tight} do not seem to work here. Via an indirect argument, in Lemma~\ref{lem:srrI} in the Appendix, we show that (\romannumeral1) the strong converse exponent in Eq.~\eqref{eq:maindec} is strictly positive if and only if $r<\frac{1}{2} I(R:A)_\rho$, and (\romannumeral2) $I_{\alpha}^{*, \rm{reg}}(R:A)_\rho\rar I(R:A)_\rho$ when $\alpha\nearrow1$. We also show in Lemma~\ref{lem:srrI} that in the optimization of Eq.~\eqref{eq:maindec}, the range $\frac{1}{2}\leq\alpha\leq 1$ can be replaced by $\frac{1}{2}<\alpha<1$, in agreement with Eq.~\eqref{eq:intdec}.

\section{Proof of the Strong Converse Exponent for Smoothing the Max-Relative Entropy}
  \label{sec:proof-mre}
In this section, we prove Theorem~\ref{theorem:dmax}. Our strategy is to first prove it in the special case that $\rho$ and $\sigma$ commute, and then reduce the general quantum case to the commutative case.

\subsection{Variational expression}
In the following proposition and throughout the paper, we use $|x|^+$ to denote $\max\{x, 0\}$.
\begin{proposition}
\label{prop:dmaxvar}
For $\rho \in \mc{S}(\mc{H})$, $\sigma \in \mc{P}(\mc{H})$ and $r \in \mathbb{R}$, we have
\begin{equation}\label{eq:dmaxvar}
\sup_{\frac{1}{2} \leq \alpha < 1} \frac{1-\alpha}{\alpha} \big\{D_{\alpha}^{\flat}(\rho \| \sigma)-r  \big\}=\inf_{\tau \in \mc{S}(\mc{H})} \left\{ D(\tau \| \rho)+\big|D(\tau \| \sigma)-r \big|^+ \right\}.
\end{equation}
Furthermore, if $\rho$ and $\sigma$ commute, then $\mc{S}(\mc{H})$ in the infimisation can be replaced by $\mc{C}_{\rho,\sigma}:=\{\tau~|~\tau\in\mc{S}(\mc{H}),~\tau~\text{commutes with}~\rho\text{~and~}\sigma\}$.
\end{proposition}
\begin{proof}
In the case that $\supp(\rho)\cap\supp(\sigma)=\{0\}$, it is easy to see that both sides of Eq.~\eqref{eq:dmaxvar} are infinity. In what follows, we prove the nontrivial case that $\supp(\rho)\cap\supp(\sigma)\neq\{0\}$. Set
\beq
\mc{S}_{\rho,\sigma}:=\{\tau\in\mc{S}(\mc{H})~|~\supp(\tau)\subseteq\supp(\rho)\cap\supp(\sigma)\}.
\eeq
We have
\begin{align}
&\sup_{\frac{1}{2} \leq \alpha < 1} \frac{1-\alpha}{\alpha}
 \left\{D_{\alpha}^{\flat}(\rho \| \sigma)-r  \right\} \nb\\
=&\sup_{\frac{1}{2} \leq \alpha < 1} \inf_{\tau \in \mc{S}(\mc{H})}\frac{1-\alpha}{\alpha}
 \left\{D(\tau \| \sigma)-\frac{\alpha}{\alpha-1}D(\tau \|\rho)-r \right\} \nb\\
=&\sup_{\frac{1}{2} \leq \alpha < 1} \inf_{\tau \in \mc{S}(\mc{H})}
 \left\{D(\tau \| \rho)+\frac{1-\alpha}{\alpha}(D(\tau \| \sigma)-r) \right\} \nb\\
=&\sup_{0 < \lambda  \leq 1} \inf_{\tau \in \mc{S}_{\rho,\sigma}}
 \left\{D(\tau \| \rho)+\lambda(D(\tau \| \sigma)-r) \right\} \nb\\
=&\inf_{\tau \in \mc{S}_{\rho,\sigma}}\sup_{0 < \lambda \leq 1}
 \left\{D(\tau \| \rho)+\lambda(D(\tau \| \sigma)-r) \right\} \nb\\
=&\inf_{\tau \in \mc{S}_{\rho,\sigma}}\left\{D(\tau\|\rho)+\big|D(\tau\|\sigma)-r\big|^+\right\} \nb\\
=&\inf_{\tau \in \mc{S}(\mc{H})}\left\{D(\tau\|\rho)+\big|D(\tau\|\sigma)-r\big|^+\right\},
\end{align}
where the second line is by the variational expression for $D_{\alpha}^{\flat}(\rho \| \sigma)$ (see Proposition~\ref{prop:mainpro}~(\romannumeral3)), and for the fifth line, we use Sion's minimax theorem (Lemma~\ref{lem:minimax} in the Appendix).

When $\rho$ and $\sigma$ commute, they have a common eigenbasis. Denote it as $\{\ket{x}\}_{x\in\mc{X}}$. For any $\tau \in \mc{S}(\mc{H})$, we write
\beq
\tilde{\tau}=\sum_{x\in\mc{X}}\bra{x}\tau\ket{x}\proj{x}.
\eeq
Then we have $\tilde{\tau}\in\mc{C}_{\rho,\sigma}$, and, due to the data processing inequality
of the relative entropy,
\beq
D(\tilde{\tau}\|\rho)+\big|D(\tilde{\tau}\|\sigma)-r\big|^+\leq D(\tau\|\rho)+\big|D(\tau\|\sigma)-r\big|^+.
\eeq
So in this case the infimisation in Eq.~\eqref{eq:dmaxvar} can be restricted to $\tau\in\mc{C}_{\rho,\sigma}$.
\end{proof}

\begin{remark}\label{rk:vd}
When $\rho$ and $\sigma$ commute, we have $D_{\alpha}^{\flat}(\rho \| \sigma)=D_{\alpha}^{*}(\rho \| \sigma)=\frac{1}{\alpha-1} \log \tr(\rho^\alpha\sigma^{1-\alpha})$, and Proposition~\ref{prop:dmaxvar} can be equivalently stated as follows. For any probability distribution $p\in\mc{Q}(\mc{X})$, nonnegative function $q$ on $\mc{X}$, and $r \in \mathbb{R}$, we have
\begin{equation}
\sup_{\frac{1}{2} \leq \alpha < 1} \frac{1-\alpha}{\alpha} \big\{D_{\alpha}(p \| q)-r \big\}=\inf_{t \in \mc{Q}(\mc{X})} \left\{ D(t \| p)+\big|D(t \| q)-r \big|^+ \right\},
\end{equation}
where $D_{\alpha}(p \| q)=\frac{1}{\alpha-1} \log\sum_{x\in\mc{X}}p(x)^\alpha q(x)^{1-\alpha}$ is the classical R\'enyi relative entropy.
\end{remark}

\subsection{Commutative case}
\begin{theorem}
\label{thm:commu}
Let $\rho \in \mc{S}(\mc{H})$, $\sigma \in \mc{P}(\mc{H})$
and $r \in \mathbb{R}$. Suppose that $\rho$ and $\sigma$ commute. We have
\begin{equation}
\lim_{n \rar \infty} \frac{-1}{n} \log \big(1-\epsilon(\rho^{\ox n} \| \sigma^{\ox n}, nr) \big)
=\sup_{\frac{1}{2} \leq \alpha < 1} \frac{1-\alpha}{\alpha} \{ D^*_{\alpha}(\rho \| \sigma) -r \}.
\end{equation}
\end{theorem}
\begin{proof}
From the definition of $\epsilon(\rho^{\ox n} \| \sigma^{\ox n}, nr)$, we see that it suffices to consider the exponential rate of decay of the following quantity:
\begin{equation}
\label{eq:equiva}
A_n:=\max \{F(\rho^{\ox n}, \tilde{\rho}_n)~|~\tilde{\rho}_n \in\mc{S}_\leq(\mc{H}^{\ox n}),
~\tilde{\rho}_n \leq 2^{nr} \sigma^{\ox n}\}.
\end{equation}
Because $\rho$ and $\sigma$ commute, there exists an orthonormal basis $\{\ket{x}\}_{x \in \mc{X}}$ of $\mc{H}$ such that $\rho$ and $\sigma$ can be simultaneously diagonalized, i.e.,
\beq
\rho=\sum_{x \in \mc{X}} p(x) \proj{x}, \quad \sigma=\sum_{x \in \mc{X}} q(x)\proj{x}.
\eeq
Correspondingly, $\rho^{\ox n}$ and $\sigma^{\ox n}$ can be written as
\beq
\rho^{\ox n}=\sum_{x^n \in \mc{X}^{\times n}} p^n(x^n)\proj{x^n}, \quad
\sigma^{\ox n}=\sum_{x^n \in \mc{X}^{\times n}} q^n(x^n)\proj{x^n},
\eeq
where $p^n(x^n)=\prod_{i=1}^np(x_i)$, $q^n(x^n)=\prod_{i=1}^nq(x_i)$ and
$\ket{x^n}=\bigotimes_{i=1}^n\ket{x_i}$. We will employ the method of types. For
a type $t\in\mc{T}^\mc{X}_n$, let $P_n(t)=\sum_{x^n\in T_n^t}p^n(x^n)$ be the
total weight of the distribution $p^n$ falling in the type class $T_n^t$, and
similarly let $Q_n(t)=\sum_{x^n\in T_n^t}q^n(x^n)$. Then we can write
\beq
\label{eq:stateform}
\rho^{\ox n}=\sum_{t \in \mc{T}^{\mc{X}}_n} P_n(t) \pi_n(t), \quad
\sigma^{\ox n}=\sum_{t \in \mc{T}^{\mc{X}}_n} Q_n(t) \pi_n(t),
\eeq
where $\pi_n(t)=\frac{1}{|T_n^t|}\sum_{x^n\in T_n^t}\proj{x_n}$ is the maximally
mixed state on the subspace $\mc{H}_n^t:=\Span\{\ket{x^n}~|~x^n\in T_n^t \}$.
Note that $\mc{H}^{\ox n}=\bigoplus_{t \in \mc{T}_n^{\mc{X}}} \mc{H}_n^t$
is a decomposition of $\mc{H}^{\ox n}$ into the direct sum of orthogonal
subspaces.

We now show that the state $\tilde{\rho}_n$ in the maximization in
Eq.~(\ref{eq:equiva}) can be restricted to the following form:
\beq\label{eq:optate}
\tilde{\rho}_n= \sum_{t \in \mc{T}_n^{\mc{X}}} \tilde{P}_n(t)\pi_n(t), \text{~~with~~}
\sum_{t \in \mc{T}_n^{\mc{X}}}\tilde{P}_n(t) \leq 1.
\eeq
Let $\Pi_n^t$ be the projection onto the subspace $\mc{H}_n^t$. We construct the
following CPTP map:
\beq
\mc{M}_n:\ X\ \ \mapsto\ \ \sum_{t \in \mc{T}_n^{\mc{X}}} \tr(\Pi_n^tX\Pi_n^t)\pi_n(t).
\eeq
Obviously, $\mc{M}_n(\rho^{\ox n})=\rho^{\ox n}$ and $\mc{M}_n(\sigma^{\ox n})=\sigma^{\ox n}$.
Thus, for any state $\tilde{\rho}_n \in\mc{S}_\leq(\mc{H}^{\ox n})$ such that
$\tilde{\rho}_n \leq 2^{nr} \sigma^{\ox n}$, we have
\begin{align}
&F(\rho^{\ox n}, \mc{M}_n(\tilde{\rho}_n)) \geq F(\rho^{\ox n}, \tilde{\rho}_n), \label{eq:type1}\\
&\mc{M}_n(\tilde{\rho}_n) \leq 2^{nr} \sigma^{\ox n} \label{eq:type2},
\end{align}
where the first inequality is because the fidelity is non-decreasing under CPTP maps,
and the second one is because $\mc{M}_n$ keeps positivity. So, if $\tilde{\rho}_n$
is not of the form as in Eq.~\eqref{eq:optate}, an action of $\mc{M}_n$ will make it
so, and Eqs.~\eqref{eq:type1} and \eqref{eq:type2} ensure that we can restrict the state
$\tilde{\rho}_n$ in the maximization of Eq.~(\ref{eq:equiva}) to be of such a form.

With Eqs.~\eqref{eq:stateform} and \eqref{eq:optate}, Eq.~(\ref{eq:equiva}) translates to
\begin{equation}
\label{eq:optimal}
A_n=\max \left\{\sum_{t \in \mc{T}_n^\mc{X}} \sqrt{P_n(t)\tilde{P}_n(t)}
~\Big|~\sum_{t \in \mc{T}^{\mc{X}}_n} \tilde{P}_n(t)\leq 1,
~(\forall t \in \mc{T}_n^\mc{X})~\tilde{P}_n(t) \leq 2^{nr} Q_n(t) \right\}.
\end{equation}
Let
\begin{align}
B_n(t)
&= \max \left\{\sqrt{P_n(t)\tilde{P}_n(t)}~\Big|~\tilde{P}_n(t)\leq\min\{2^{nr}Q_n(t),1\}\right\}\nb\\
&=\sqrt{P_n(t)}\sqrt{\min\{2^{nr}Q_n(t), 1 \}}
\end{align}
Then, we have
\begin{equation}
\label{eq:relation}
\max_{t \in \mc{T}^{\mc{X}}_n} B_n(t) \leq  A_n \leq (n+1)^{|\mc{X}|} \max_{t \in \mc{T}^{\mc{X}}_n} B_n(t),
\end{equation}
where the second inequality is due to Eq.~(\ref{eq:typenumber}).
The remaining thing is to estimate $\max_{t \in \mc{T}_n^{\mc{X}}} B_n(t)$. Using Eq.~(\ref{eq:proset}), we bound it as
\begin{align}
\label{eq:Bnt1}
      \max_{t \in \mc{T}^{\mc{X}}_n} B_n(t)
&\geq \max_{t \in \mc{T}^{\mc{X}}_n} \left\{\sqrt{(n+1)^{-|\mc{X}|}
      2^{-nD(t \| p)}}\ \min\Big\{\sqrt{(n+1)^{-|\mc{X}|}
      2^{-n(D(t \| q)-r)}}, 1 \Big\} \right\},  \\
\label{eq:Bnt2}
      \max_{t \in \mc{T}^{\mc{X}}_n} B_n(t)
&\leq \max_{t \in \mc{T}^{\mc{X}}_n} \left\{\sqrt{2^{-nD(t \| p)}}\
      \min\left\{ \sqrt{2^{-n(D(t \| q)-r)}}, 1 \right\} \right\}.
\end{align}
Eq.~(\ref{eq:relation}), Eq.~(\ref{eq:Bnt1}) and Eq.~(\ref{eq:Bnt2}) give
\begin{align}
\frac{-1}{n} \log A_n
&\geq \frac{1}{2}\min_{t \in \mc{T}^{\mc{X}}_n}
      \left\{D(t \| p)+\big|D(t \| q)-r\big|^+\right\} -|\mc{X}| \frac{\log (n+1)}{n}, \label{eq:Ant11} \\
\frac{-1}{n} \log A_n
&\leq \frac{1}{2}\min_{t \in \mc{T}^{\mc{X}}_n}
      \left\{D(t \| p)+\Big|D(t \| q)-r+|\mc{X}| \frac{\log (n+1)}{n}\Big|^+ \right\}
      +\frac{1}{2}|\mc{X}| \frac{\log (n+1)}{n}  \nb\\
&\leq \frac{1}{2}\min_{t \in \mc{T}^{\mc{X}}_n}
      \left\{D(t \| p)+\big|D(t \| q)-r\big|^+\right\} +|\mc{X}| \frac{\log (n+1)}{n}. \label{eq:Ant22}
\end{align}
For a fixed nonnegative function $b$ on $\mc{X}$, the function $a\mapsto D(a\|b)$ defined on $\mc{Q}(\mc{X})$ is $+\infty$ when $\supp(a)\nsubseteq\supp(b)$ and otherwise it is uniformly continuous. On the other hand, $\mc{T}^{\mc{X}}_n$ is dense in the probability simplex $\mc{Q}(\mc{X})$ as $n \rar \infty$. So, by taking the limit, we get from Eqs.~\eqref{eq:Ant11} and \eqref{eq:Ant22} that
\begin{align}
\lim_{n \rar \infty} \frac{-1}{n} \log A_n
&=\frac{1}{2}\min_{t \in \mc{Q}(\mc{X})} \left\{ D(t \| p)+\big|D(t \| q)-r\big|^+\right\} \nb\\
&=\frac{1}{2}\min_{\tau \in \mc{C}_{\rho,\sigma}} \left\{ D(\tau \| \rho)
    +\big|D(\tau \| \sigma)-r\big|^+\right\} \nb\\
&=\frac{1}{2}\sup_{\frac{1}{2} \leq \alpha < 1} \frac{1-\alpha}{\alpha}
    \big\{D_{\alpha}^*(\rho \| \sigma)-r\big\}, \label{eq:conclusion}
\end{align}
where the last line is by Proposition~\ref{prop:dmaxvar} and Remark~\ref{rk:vd}. Note that for commutative $\rho$ and $\sigma$ we have $D_{\alpha}^{\flat}(\rho \| \sigma)=D_{\alpha}^*(\rho \| \sigma)$. At last, making use of the relation
\beq
1-P=1-\sqrt{1-F^2}=\frac{1}{2}F^2+o(F^2)
\eeq
between the purified distance and the fidelity, we obtain the claimed result of
Theorem~\ref{thm:commu} from Eq.~\eqref{eq:conclusion}.
\end{proof}

\subsection{General case}
Now, we are ready to complete the proof of Theorem~\ref{theorem:dmax}. The following Proposition~\ref{prop:dmaxpinching} allows us to lift the above result of the commutative case to deal with the general case. At first, we prove a technical result of Lemma~\ref{lem:likeUlman}, which is mainly due to Uhlmann's theorem.

\begin{lemma}
\label{lem:likeUlman}
Let $\mc{H}=\bigoplus_i\mc{H}_i$ be a decomposition of $\mc{H}$ into mutually orthogonal subspaces.
Let $\rho, \sigma \in \mc{S}_{\leq}(\mc{H})$ be such that $\rho=\sum_i \rho_i$ and
$\sigma=\sum_i \sigma_i$, with $\supp (\rho_i)\subseteq \mc{H}_i$ and $\supp (\sigma_i) \subseteq \mc{H}_i$. Denote the projection onto $\mc{H}_i$ as $\Pi_i$. Then for any state
$\tilde{\rho} \in \mc{S}_{\leq}(\mc{H})$ satisfying $\rho=\sum_i \Pi_i
\tilde{\rho} \Pi_i$, there exists a state $\tilde{\sigma} \in \mc{S}_{\leq}(\mc{H})$
such that
\beq
\sigma=\sum_i \Pi_i \tilde{\sigma} \Pi_i \text{\quad~and~\quad} F(\rho, \sigma)=F(\tilde{\rho}, \tilde{\sigma}).
\eeq
\end{lemma}
\begin{proof}
Let $\bar{F}(X,Y):=\|\sqrt{X}\sqrt{Y}\|_1$ for any $X,Y\in\mc{P}(\mc{H})$. By direct calculation, we see that
\begin{equation}
\label{eq:lem2-1}
F(\rho, \sigma)-\sqrt{(1-\tr\rho)(1-\tr\sigma)}
=\sum_i \bar{F}(\Pi_i \tilde{\rho} \Pi_i, \sigma_i)=\sum_i \bar{F}(\tilde{\rho}, \sigma_i).
\end{equation}
Let $\ket{\psi} \in \mc{H} \ox \mc{H}'$ be a purification of $\tilde{\rho}$, where $\mc{H} \cong\mc{H}'$ (that is, $\tilde{\rho}=\tr_{\mc{H}'}\proj{\psi}$). By Uhlmann's theorem~\cite{Uhlmann1976transition}, for each $\sigma_i$, there exists a purification $\{ \ket{\phi_i} \} \in \mc{H} \ox \mc{H}'$, such that
\begin{equation}
\label{eq:lem2-2}
\bar{F}(\tilde{\rho}, \sigma_i)=\braket{\psi}{\phi_i}.
\end{equation}
Denote $\ket{\phi}=\sum_i \ket{\phi_i}$, we will show that $\tilde{\sigma}:=\tr_{\mc{H}'}\proj{\phi}$ satisfies our requirement. Since $\Pi_i\ket{\phi}=\ket{\phi_i}$, we can easily check that $\sigma=\sum_i \Pi_i \tilde{\sigma} \Pi_i$. So it remains to show $F(\rho, \sigma)=F(\tilde{\rho}, \tilde{\sigma})$. Eq.~(\ref{eq:lem2-1}) and Eq.~(\ref{eq:lem2-2}) give that
\begin{align}
F(\rho, \sigma) &= \sum_i \braket{\psi}{\phi_i}+\sqrt{(1-\tr\rho)(1-\tr\sigma)} \nb\\
&=\bar{F}(\proj{\psi}, \proj{\phi})+\sqrt{(1-\tr\proj{\psi})(1-\tr\proj{\phi})} \nb\\
&=F(\proj{\psi}, \proj{\phi}). \label{eq:purification}
\end{align}
On the other hand, by the monotonicity of fidelity under the action of quantum channels, we have
\beq\label{eq:monofidelity}
F(\rho, \sigma)\geq F(\tilde{\rho}, \tilde{\sigma}) \geq F(\proj{\psi}, \proj{\phi}).
\eeq
Eq.~\eqref{eq:purification} and Eq.~\eqref{eq:monofidelity} lead to $F(\rho, \sigma)=F(\tilde{\rho}, \tilde{\sigma})$, completing the proof.
\end{proof}

\begin{proposition}
\label{prop:dmaxpinching}
Let $\mc{H}=\bigoplus_{i \in \mc{I}}\mc{H}_i$ be a decomposition of $\mc{H}$ into mutually
orthogonal subspaces. Let $\rho \in \mc{S}(\mc{H})$, and $\sigma \in \mc{P}(\mc{H})$
be such that $\sigma=\sum_{i \in \mc{I}} \sigma_i$ with $\supp (\sigma_i) \subseteq \mc{H}_i$. Denote the projection onto $\mc{H}_i$ as $\Pi_i$ and let $\mc{E}$ be the pinching map given by $\mc{E}(X)=\sum_{i \in \mc{I}} \Pi_i X \Pi_i$. Then for any $0\leq\epsilon\leq 1$ we have
\begin{equation}
D_{\rm{max}}^{\epsilon}(\mc{E}(\rho) \| \sigma) \leq D_{\rm{max}}^{\epsilon}(\rho \| \sigma)
\leq D_{\rm{max}}^{\epsilon}(\mc{E}(\rho) \| \sigma)+\log |\mc{I}|.
\end{equation}
\end{proposition}
\begin{proof}
The first inequality is known in the literature (see, e.g., \cite{Renner2005security,Tomamichel2015quantum}), and can be verified directly from the definition. We now prove the second inequality. It holds trivially if $D_{\rm{max}}^{\epsilon}(\mc{E}(\rho) \| \sigma)=+\infty$. So, we assume that $D_{\rm{max}}^{\epsilon}(\mc{E}(\rho) \| \sigma)<+\infty$. In this case, by the definition of the smooth max-relative entropy, we know that there exists a state $\bar{\rho}\in\mc{S}_{\leq}(\mc{H})$ such that
\begin{equation}
\label{eq:definitionmax}
P(\mc{E}(\rho),\bar{\rho}) \leq \epsilon, \quad \bar{\rho} \leq 2^{D_{\rm{max}}^{\epsilon}(\mc{E}(\rho) \| \sigma)}\sigma.
\end{equation}
Without loss of generality, we assume that $\bar{\rho}$ can be written as $\bar{\rho}=\sum_{i \in \mc{I}} \bar{\rho}_i$ with $\supp (\bar{\rho}_i) \subseteq \mc{H}_i$. Otherwise, applying $\mc{E}$ to $\bar{\rho}$ will make it so and still satisfy the above inequalities.
Then, according to Lemma~\ref{lem:likeUlman}, there exists a state $\tilde{\rho}\in\mc{S}_{\leq}(\mc{H})$ such that
\begin{equation}
\label{eq:applylemma}
\bar{\rho}= \mc{E}(\tilde{\rho}),\quad F(\rho,\tilde{\rho})=F(\mc{E}(\rho),\bar{\rho}).
\end{equation}
The combination of Eq.~(\ref{eq:definitionmax}) and Eq.~(\ref{eq:applylemma}) leads to
\begin{equation}
\label{eq:dmaxpp}
P(\rho,\tilde{\rho})=P(\mc{E}(\rho),\bar{\rho}) \leq \epsilon
\end{equation}
and
\begin{equation}
\label{eq:dmaxpd}
\tilde{\rho} \leq |\mc{I}|\mc{E}(\tilde{\rho})
\leq 2^{D_{\rm{max}}^{\epsilon}(\mc{E}(\rho) \| \sigma)+\log |\mc{I}|}\sigma,
\end{equation}
where for the latter inequality we have used the pinching inequality.
Eq.~(\ref{eq:dmaxpp}) and Eq.~(\ref{eq:dmaxpd}) imply that
\begin{equation}
D_{\rm{max}}^{\epsilon}(\rho \| \sigma)
\leq D_{\rm{max}}^{\epsilon}(\mc{E}(\rho) \| \sigma)+\log |\mc{I}|,
\end{equation}
and we are done.
\end{proof}

\begin{remark}
\label{remark:smoothqpinching}
An equivalent statement of Proposition~\ref{prop:dmaxpinching} is that, under the same conditions,
for any $\lambda\in\mathbb{R}$ we have
\beq
\label{eq:smoothqpinching}
    \epsilon(\mc{E}(\rho) \| \sigma, \lambda)
\leq\epsilon(\rho \| \sigma, \lambda)
\leq\epsilon(\mc{E}(\rho) \| \sigma, \lambda-\log|\mc{I}|).
\eeq
\end{remark}

\begin{remark}
\label{remark:dmaxpinching}
A special case of Proposition~\ref{prop:dmaxpinching} which may be of particular interest, is when $\mc{E}=\mc{E}_\sigma$ being the pinching map associated with $\sigma$. In this case, for any $0\leq\epsilon\leq 1$ we have
\begin{equation}
D_{\rm{max}}^{\epsilon}(\mc{E}_\sigma(\rho) \| \sigma) \leq D_{\rm{max}}^{\epsilon}(\rho \| \sigma)
\leq D_{\rm{max}}^{\epsilon}(\mc{E}_\sigma(\rho) \| \sigma)+\log |v(\sigma)|.
\end{equation}
\end{remark}

\medskip
We also need two simple properties of the smoothing quantity $\epsilon(\rho \| \sigma, \lambda)$, which are stated in the following Lemma~\ref{lem:smoothq}, and can be verified directly from the definition.
\begin{lemma}
\label{lem:smoothq}
Let $\rho \in \mc{S}(\mc{H})$ and $\sigma \in \mc{P}(\mc{H})$. we have
\begin{enumerate}[(i)]
  \item Monotonicity: the function $\lambda\mapsto\epsilon(\rho \| \sigma, \lambda)$ is monotonically non-increasing;
  \item Data processing inequality: $\epsilon(\mc{N}(\rho) \| \mc{N}(\sigma), \lambda) \leq \epsilon(\rho \| \sigma, \lambda)$ holds for any CPTP map $\mc{N}$ and any $\lambda\in\mathbb{R}$.
\end{enumerate}
\end{lemma}

\medskip
\begin{proofof}[of Theorem \ref{theorem:dmax}]
Let $\delta>0$ be arbitrary and let $M_{\delta}$ be an integer such that $\frac{\log v(\sigma^{\ox m})}{m} \leq \delta$ for all $m \geq M_\delta$. Fix such an integer $m$. Then for any integer $n$ we write $n=km+l$ with $0\leq l <m$. On the one hand, for $n$ big enough such that $nr\leq km(r+\delta)$, by Lemma~\ref{lem:smoothq} we have
\begin{align}
      \epsilon\big(\rho^{\ox n} \| \sigma^{\ox n}, nr\big)
&\geq \epsilon\left(\left(\rho^{\ox m}\right)^{\ox k} \big\|
              \left(\sigma^{\ox m}\right)^{\ox k}, km(r+\delta)\right) \nb\\
&\geq \epsilon\left(\left(\mc{E}_{\sigma^{\ox m}}(\rho^{\ox m})\right)^{\ox k}\big\|
              \left(\sigma^{\ox m}\right)^{\ox k}, km(r+\delta)\right). \label{eq:dmaxf11}
\end{align}
On the other hand, for $n$ big enough such that $nr\geq (k+1)m(r-\delta)$, we also have
\begin{align}
      \epsilon\big(\rho^{\ox n} \| \sigma^{\ox n}, nr\big)
&\leq \epsilon\left(\left(\rho^{\ox m}\right)^{\ox (k+1)} \big\|
              \left(\sigma^{\ox m}\right)^{\ox (k+1)}, (k+1)m(r-\delta)\right) \nb\\
&\leq \epsilon\left(\left(\mc{E}_{\sigma^{\ox m}}(\rho^{\ox m})\right)^{\ox (k+1)}\big\|
              \left(\sigma^{\ox m}\right)^{\ox (k+1)}, (k+1)m(r-2\delta)\right), \label{eq:dmaxf12}
\end{align}
where the first inequality is by Lemma~\ref{lem:smoothq} and the second inequality follows
from Proposition~\ref{prop:dmaxpinching} (in the form of Eq.~\eqref{eq:smoothqpinching}). Because $\mc{E}_{\sigma^{\ox m}}(\rho^{\ox m})$ commutes with $\sigma^{\ox m}$, Theorem~\ref{thm:commu} applies to the last lines of Eq.~\eqref{eq:dmaxf11} and Eq.~\eqref{eq:dmaxf12}. This leads to
\begin{align}
&\liminf_{n \rar \infty} \frac{-1}{n}\log\big(1-\epsilon(\rho^{\ox n} \| \sigma^{\ox n}, nr)\big) \geq \sup_{\frac{1}{2} \leq \alpha < 1} \frac{1-\alpha}{\alpha} \left\{\frac{D_{\alpha}^*(\mc{E}_{\sigma^{\ox m}}(\rho^{\ox m}) \| \sigma^{\ox m})}{m} -r-\delta \right\}, \label{eq:dmaxf21} \\
&\limsup_{n \rar \infty}\!\frac{-1}{n}\log\big(1-\epsilon(\rho^{\ox n} \| \sigma^{\ox n}, nr)\big) \leq\! \sup_{\frac{1}{2} \leq \alpha < 1} \frac{1-\alpha}{\alpha} \left\{\!\frac{D_{\alpha}^*(\mc{E}_{\sigma^{\ox m}}(\rho^{\ox m}) \| \sigma^{\ox m})}{m} -r+2\delta\! \right\}. \label{eq:dmaxf22}
\end{align}
Since $\delta>0$ and $m\geq M_{\delta}$ are arbitrary, first letting $m\rar\infty$ and then letting $\delta\rar 0$ in Eq.~\eqref{eq:dmaxf21} and Eq.~\eqref{eq:dmaxf22}, and also making use of the result proved in~\cite[Lemma 3]{HayashiTomamichel2016correlation} that
\begin{equation}
\label{eq:dmaxf3}
\frac{1}{m}D_{\alpha}^*(\mc{E}_{\sigma^{\ox m}}(\rho^{\ox m}) \| \sigma^{\ox m})
\stackrel{m\rar\infty}{\lrar} D^{*}_{\alpha}(\rho \| \sigma)
\end{equation}
uniformly for $\alpha\in[\frac{1}{2},1)$, we eventually arrive at
\begin{equation}
\lim_{n \rar \infty} \frac{-1}{n} \log \big(1-\epsilon(\rho^{\ox n} \| \sigma^{\ox n}, nr) \big)
=\sup_{\frac{1}{2} \leq \alpha < 1} \frac{1-\alpha}{\alpha} \left\{ D^*_{\alpha}(\rho \| \sigma) -r \right\}
\end{equation}
and we are done.
\end{proofof}

\section{Proof of the Strong Converse Exponent for Quantum Privacy Amplification}
  \label{sec:proof-qpa}
This section is devoted to the proof of Theorem~\ref{theorem:privacy}. The first three subsections are for the achievability part, and the last subsection is for the optimality part.

\subsection{Variational expression}
We define an intermediate quantity $G(\rho_{XE}, r)$, which will be proven to be an upper bound of $E_{\rm{sc}}^{\rm{pa}}(\rho_{XE}, r)$.
\begin{definition}
Let $\rho_{XE} \in \mc{S}(XE)$ be a CQ state and $r \geq 0$. The quantity $G(\rho_{XE}, r)$ is defined as
\begin{equation}
G(\rho_{XE}, r):= \sup_{\frac{1}{2} \leq \alpha \leq 1} \frac{1-\alpha}{\alpha}\left\{r-H_{\alpha}^{\flat}(X|E)_\rho\right\}.
\end{equation}
\end{definition}

We will prove a variational expression for $G(\rho_{XE}, r)$. The following Proposition~\ref{prop:varpa}
provides a slightly general result. The exact variational expression for $G(\rho_{XE}, r)$ is given
in Eq.~\eqref{eq:varG2}.

\begin{proposition}
\label{prop:varpa}
Let $\rho_{XE}\in\mc{S}(XE)$ and $r \geq 0$. We have
\begin{equation}\label{eq:varG1}
\sup_{\frac{1}{2} \leq \alpha \leq 1} \frac{1-\alpha}{\alpha}\left\{r-H_{\alpha}^{\flat}(X|E)_\rho\right\}
=\inf_{\tau_{XE} \in \mc{S}_{\rho}(XE)} \left\{D(\tau_{XE} \| \rho_{XE})+|r-H(X|E)_\tau|^+ \right\}.
\end{equation}
Furthermore, if $\rho_{XE}$ is a CQ state, then the set $\mc{S}_{\rho}(XE)$ in the infimisation can be replaced by $\mc{F}:=\{\tau_{XE}~|~\tau_{XE}~\text{is a CQ state},~\tau_{XE} \in \mc{S}_{\rho}(XE)\}$. That is,
\begin{equation}\label{eq:varG2}
G(\rho_{XE}, r)
=\inf_{\tau_{XE} \in \mc{F}} \left\{D(\tau_{XE} \| \rho_{XE})+|r-H(X|E)_\tau|^+ \right\}.
\end{equation}
\end{proposition}

\begin{proof}
When $\alpha\in[\frac{1}{2},1)$, we have
\begin{align}
 & \frac{1-\alpha}{\alpha}\left\{r-H_{\alpha}^{\flat}(X|E)_{\rho}\right\} \nb\\
=& \frac{1-\alpha}{\alpha}\left\{\inf_{\sigma_E\in\mc{S}(E)}
   D_{\alpha}^{\flat}(\rho_{XE} \| \1_X \ox \sigma_E)+r\right\} \nb\\
=& \inf_{\sigma_E\in\mc{S}(E)} \inf_{\tau_{XE} \in \mc{S}(XE)}\frac{1-\alpha}{\alpha}
   \left\{D(\tau_{XE}\|{\1}_X\ox\sigma_E)-\frac{\alpha}{\alpha-1}D(\tau_{XE}\|\rho_{XE})+r\right\} \nb\\
=& \inf_{\tau_{XE}\in\mc{S}(XE)}\left\{\frac{1-\alpha}{\alpha}\big(D(\tau_{XE}\|\1_X\ox\tau_E)+r\big)
   +D(\tau_{XE} \| \rho_{XE}) \right\}, \nb\\
=& \inf_{\tau_{XE}\in\mc{S}_\rho(XE)}\left\{\frac{1-\alpha}{\alpha}\big(r-H(X|E)_{\tau}\big)
   +D(\tau_{XE} \| \rho_{XE}) \right\}, \label{eq:v2}
\end{align}
where the second line is by definition and the third line is by the variational expression for $D_{\alpha}^{\flat}(\rho \| \sigma)$ (see Proposition~\ref{prop:mainpro}~(\romannumeral3)). When $\alpha=1$, the first line and last line of Eq.~\eqref{eq:v2} are still equal, because they are both $0$. So,
\begin{align}
&\sup_{\frac{1}{2} \leq \alpha \leq 1}
 \frac{1-\alpha}{\alpha}\left\{r-H_{\alpha}^{\flat}(X|E)_{\rho}\right\} \nb\\
=&\sup_{\frac{1}{2} \leq \alpha \leq 1}\inf_{\tau_{XE}\in\mc{S}_\rho(XE)}
  \left\{\frac{1-\alpha}{\alpha}\big(r-H(X|E)_{\tau}\big)+D(\tau_{XE} \| \rho_{XE}) \right\} \nb\\
=&\sup_{\lambda \in [0,1]} \inf_{\tau_{XE} \in \mc{S}_{\rho}(XE)}
  \left\{\lambda\big(r-H(X|E)_{\tau}\big)+D(\tau_{XE} \| \rho_{XE}) \right\} \nb\\
=&\inf_{\tau_{XE} \in \mc{S}_{\rho}(XE)} \sup_{\lambda \in [0,1]}
  \left\{\lambda\big(r-H(X|E)_{\tau}\big)+D(\tau_{XE} \| \rho_{XE}) \right\}  \nb\\
=&\inf_{\tau_{XE}\in\mc{S}_{\rho}(XE)}\left\{D(\tau_{XE} \| \rho_{XE})+|r-H(X|E)_{\tau}|^+\right\},
\end{align}
where for the fourth line we have used Sion's minimax theorem (Lemma~\ref{lem:minimax} in the Appendix). Note that the functions $\tau_{XE}\mapsto -H(X|E)_{\tau}$ and $\tau_{XE}\mapsto D(\tau_{XE} \| \rho_{XE})$ are both convex (see, e.g., \cite{Wilde2013quantum}). This confirms Eq.~\eqref{eq:varG1}.

Now we consider the special case that $\rho_{XE}=\sum_{x \in \mc{X}} q(x)\proj{x}_X \ox \rho_E^x$ is a CQ state. For any $\tau_{XE}\in\mc{S}_{\rho}(XE)$, let
\beq
\tilde{\tau}_{XE}:=\sum_{x \in \mc{X}}\proj{x} \ox \tr_X[\tau_{XE} (\proj{x}_X\ox \1_E)].
\eeq
Obviously, $\tilde{\tau}_{XE}$ is a CQ state and, by the data processing inequality of the relative entropy,
\begin{align}
D(\tilde{\tau}_{XE} \| \rho_{XE}) &\leq D(\tau_{XE} \| \rho_{XE}), \\
H(X|E)_{\tilde{\tau}} &\geq H(X|E)_{\tau}.
\end{align}
So, in this case, the infimisation in the right hand side of Eq.~\eqref{eq:varG1} can be restricted to $\tau_{XE} \in \mc{F}$, and we get Eq.~\eqref{eq:varG2}.
\end{proof}

\subsection{An upper bound in terms of the log-Euclidean R{\'e}nyi conditional entropy}
Now we show that $G(\rho_{XE}, r)$ is achievable as a strong converse rate for privacy amplification.
\begin{theorem}
\label{thm:PA-F}
Let $\rho_{XE}$ be a CQ state and $r \geq 0$. It holds that
\begin{equation}
E_{\rm{sc}}^{\rm{pa}}(\rho_{XE}, r) \leq  G(\rho_{XE}, r).
\end{equation}
\end{theorem}
To prove Theorem~\ref{thm:PA-F}, we introduce
\begin{align}
G_1(\rho_{XE}, r)&:=\inf_{\tau_{XE} \in \mc{F}_1} D(\tau_{XE} \| \rho_{XE}),\\
G_2(\rho_{XE}, r)&:=\inf_{\tau_{XE} \in \mc{F}_2} \{D(\tau_{XE} \| \rho_{XE})+r-H(X|E)_{\tau} \},
\end{align}
where
\begin{align}
\mc{F}_1&:=\{\tau_{XE}~|~\tau_{XE} \in \mc{F},~r<H(X|E)_{\tau} \},\\
\mc{F}_2&:=\{ \tau_{XE}~|~\tau_{XE} \in \mc{F},~r \geq H(X|E)_{\tau}\}
\end{align}
with $\mc{F}$ being defined in Proposition~\ref{prop:varpa}. Due to the variational expression of Eq.~\eqref{eq:varG2}, we have
\begin{equation}
\label{eq:FF_1}
G(\rho_{XE}, r)=\min \{G_1(\rho_{XE}, r),  G_2(\rho_{XE}, r) \}.
\end{equation}
Hence, it suffices to show that $E_{\rm{sc}}^{\rm{pa}}(\rho_{XE}, r)$ is upper bounded by both $G_1(\rho_{XE}, r)$ and $G_2(\rho_{XE}, r)$. This is done in the following
Lemma~\ref{lem:F1} and Lemma~\ref{lem:F_2}, respectively.

\begin{lemma}
\label{lem:F1}
For any CQ state $\rho_{XE}$ and $r \geq 0$, we have
\begin{equation}
E_{\rm{sc}}^{\rm{pa}}(\rho_{XE}, r) \leq G_1(\rho_{XE}, r).
\end{equation}
\end{lemma}
\begin{proof}
By the definition of $G_1(\rho_{XE}, r)$, for any $\delta>0$, there exists a CQ state $\tau_{XE}\in\mc{S}_\rho(XE)$ such that
\begin{align}
&H(X|E)_{\tau}>r, \label{eq:priv1}\\
&D(\tau_{XE} \| \rho_{XE}) \leq G_1(\rho_{XE}, r)+\delta.\label{eq:priv2}
\end{align}
Eq.~\eqref{eq:priv1} ensures that $r$ is an achievable rate of randomness extraction from $\tau_{XE}$~\cite{DevetakWinter2005distillation}. That is, there exists a sequence of hash functions $\big\{f_n: \mc{X}^{\times n} \rar \mc{Z}_n=\{1, \ldots, 2^{nr}\}\big\}_{n \in \mathbb{N}}$, such that $\mathfrak{P}^{\rm pa}(\rho_{XE}^{\ox n},f_n)\rar 1$, as $n\rar\infty$. We point out that, due to the Fuchs-van de Graaf inequality~\cite{FuchsVan1999cryptographic} between fidelity and trace distance, as well as the triangle inequality and the contractivity property of the trace distance, this is equivalent to the original statement in~\cite{DevetakWinter2005distillation} that
\begin{equation}
\lim_{n \rar \infty}
\Big\|\mc{P}_{f_n}(\tau_{XE}^{\ox n})-\frac{\1_{Z_n}}{|\mc{Z}_n|} \ox \tau_E^{\ox n}\Big\|_1=0.
\end{equation}
In fact, \cite[Theorem 1]{Hayashi2015precise} implies that we even have the stronger form:
\begin{equation}\label{eq:paachi-1}
\lim_{n \rar \infty}
D\Big(\mc{P}_{f_n}(\tau_{XE}^{\ox n})\Big\|\frac{\1_{Z_n}}{|\mc{Z}_n|} \ox \tau_E^{\ox n}\Big)=0.
\end{equation}
To apply this result to the extraction of randomness from $\rho_{XE}$, we employ Lemma~\ref{lem:fidelity-re}. This lets us get
\begin{align}
     &-\log \mathfrak{P}^{\rm pa}\left(\rho_{XE}^{\ox n},f_n\right) \nb\\
  =  &-\log \max_{\omega_{E^n} \in \mc{S}(E^n)} F^2\Big(\mc{P}_{f_n}(\rho_{XE}^{\ox n}),
       \frac{\1_{Z_n}}{|\mc{Z}_n|} \ox \omega_{E^n}\Big) \nb\\
\leq &-\log F^2\Big(\mc{P}_{f_n}(\rho_{XE}^{\ox n}),
     \frac{\1_{Z_n}}{|\mc{Z}_n|} \ox \tau_E^{\ox n}\Big) \nb\\
\leq &D\Big(\mc{P}_{f_n}(\tau_{XE}^{\ox n}) \big\| \mc{P}_{f_n}(\rho_{XE}^{\ox n})\Big)
      +D\Big(\mc{P}_{f_n}(\tau_{XE}^{\ox n}) \big\| \frac{\1_{Z_n}}{|\mc{Z}_n|}\ox\tau_E^{\ox n}\Big) \nb\\
\leq &n D(\tau_{XE} \| \rho_{XE})+D\Big(\mc{P}_{f_n}(\tau_{XE}^{\ox n})
      \big\| \frac{\1_{Z_n}}{|\mc{Z}_n|}\ox \tau_E^{\ox n}\Big), \label{eq:paachi-2}
\end{align}
where the equality is by definition (Eq.~\eqref{eq:def-p-pa}), and the last inequality is by the data processing inequality of the relative entropy. Taking limit, and by Eq.~\eqref{eq:paachi-1}, we obtain
\beq\label{eq:paachi-3}
\limsup_{n\rar\infty}\frac{-1}{n} \log \mathfrak{P}^{\rm pa}\left(\rho_{XE}^{\ox n},f_n\right)
\leq D(\tau_{XE} \| \rho_{XE}).
\eeq
Noticing that the sequence of hash functions $\big\{f_n: \mc{X}^{\times n} \rar \mc{Z}_n\big\}_n$ satisfies
\beq
\frac{1}{n}\log |\mc{Z}_n|=r,
\eeq
we deduce from Eq.~\eqref{eq:paachi-3} and the definition of $E_{\rm{sc}}^{\rm{pa}}(\rho_{XE}, r)$ that
\begin{align}
 E_{\rm{sc}}^{\rm{pa}}(\rho_{XE}, r)
\leq &D(\tau_{XE} \| \rho_{XE}) \nb\\
\leq &G_1(\rho_{XE}, r)+\delta,
\end{align}
where the second line is due to Eq.~\eqref{eq:priv2}. Because $\delta>0$ is arbitrary, we complete the proof by letting $\delta \rar 0$.
\end{proof}

\begin{lemma}
\label{lem:F_2}
For any CQ state $\rho_{XE}$ and $r \geq 0$, we have
\begin{equation}
E_{\rm{sc}}^{\rm{pa}}(\rho_{XE}, r) \leq G_2(\rho_{XE}, r).
\end{equation}
\end{lemma}
\begin{proof}
According to the definition of $G_2(\rho_{XE}, r)$, there exists a CQ state $\tau_{XE}\in\mc{S}_\rho(XE)$ such that
\begin{align}
&D(\tau_{XE} \| \rho_{XE})+r-H(X|E)_{\tau} = G_2(\rho_{XE}, r), \\
&H(X|E)_{\tau}\leq r.
\end{align}
For any $\delta>0$, if we choose $r'=H(X|E)_{\tau}-\delta$ as the rate of randomness extraction, then from the proof of Lemma \ref{lem:F1} we know that there exist a sequence of hash functions $\big\{f'_n: \mc{X}^{\times n} \rar \mc{Z}'_n=\{1, \ldots, 2^{nr'} \}\big\}_{n \in \mathbb{N}}$ and a sequence of states $\{\sigma_{E^n} \}_{n \in \mathbb{N}}$, such that
\begin{equation}
\label{eq:above}
\limsup_{n\rar\infty} \frac{-1}{n}\log  F^2\big(\mc{P}_{f'_n}(\rho_{XE}^{\ox n}), \frac{\1_{Z'_n}}{|\mc{Z}'_n|} \ox \sigma_{E^n}\big) \leq D(\tau_{XE} \| \rho_{XE}).
\end{equation}
We convert $\{f'_n\}_{n \in \mathbb{N}}$ to a sequence of hash functions $\big\{f_n : \mc{X}^{\times n} \rar \mc{Z}_n=\{1, \ldots, 2^{nr}\}\big\}_{n\in \mathbb{N}}$, by extending the sizes of the extracted randomness but otherwise letting $f_n=f'_n$. Then we have
\begin{align}
F\big(\mc{P}_{f_n}(\rho_{XE}^{\ox n}), \frac{\1_{Z_n}}{|\mc{Z}_n|} \ox \sigma_{E^n}\big)
=&F\big(\mc{P}_{f'_n}(\rho_{XE}^{\ox n}), \frac{\1_{Z'_n}}{|\mc{Z}_n|} \ox \sigma_{E^n}\big) \nb\\
=&\sqrt{\frac{|\mc{Z}'_n|}{|\mc{Z}_n|}}F\big(\mc{P}_{f'_n}(\rho_{XE}^{\ox n}),
        \frac{\1_{Z'_n}}{|\mc{Z}'_n|} \ox \sigma_{E^n}\big) \nb\\
=&\sqrt{2^{n(r'-r)}}F\big(\mc{P}_{f'_n}(\rho_{XE}^{\ox n}),
        \frac{\1_{Z'_n}}{|\mc{Z}'_n|} \ox \sigma_{E^n}\big). \label{eq:fedl}
\end{align}
Eq.~(\ref{eq:above}) and Eq.~(\ref{eq:fedl}) imply that
\begin{align}
      E_{\rm{sc}}^{\rm{pa}}(\rho_{XE}, r)
\leq &\limsup_{n\rar\infty} \frac{-1}{n}\log F^2\big(\mc{P}_{f_n}(\rho_{XE}^{\ox n}), \frac{\1_{Z_n}}{|\mc{Z}_n|} \ox \sigma_{E^n}\big) \nb\\
\leq &D(\tau_{XE} \| \rho_{XE})+r-r' \nb\\
  =  & G_2(\rho_{XE}, r) +\delta. \label{eq:F2final}
\end{align}
Noticing that Eq.~(\ref{eq:F2final}) holds for any $\delta >0$, let $\delta \rar 0$
and we are done.
\end{proof}

\subsection{Proof of the achievability part: last step}
With Theorem~\ref{thm:PA-F}, we are now ready to complete the proof of the achievability part of Theorem~\ref{theorem:privacy}.

\medskip
\begin{proof}
At first, we fix $m$ and let $\tilde{\rho}_{X^mE^m}= \mc{E}_{\sigma^u_{E^m}}(\rho^{\ox m}_{XE})$, where $\sigma^u_{E^m}$ is the particular symmetric state of Lemma~\ref{lem:sym}.
Theorem~\ref{thm:PA-F} implies that there exist a sequence of hash functions $\{f_{m,k} : \mc{X}^{\times mk} \rar \mc{Z}_{m,k}\}_{k \in \mathbb{N}}$ and a sequence of states $\{\sigma_{E^{mk}} \}_{k \in \mathbb{N}}$, such that
\begin{equation}
\label{eq:rate}
\liminf_{k\rar\infty} \frac{1}{k} \log |\mc{Z}_{m,k}| \geq mr,
\end{equation}
\begin{equation}
\label{eq:intermit}
\limsup_{k\rar\infty} \frac{-1}{k} \log F^2\Big(\mc{P}_{f_{m,k}}\big((\tilde{\rho}_{X^mE^m})^{\ox k}\big), \frac{\1_{Z_{m,k}}}{|\mc{Z}_{m,k}|} \ox \sigma_{E^{mk}}\Big) \leq \sup_{\frac{1}{2} \leq \alpha \leq 1}
\frac{1-\alpha}{\alpha}\left\{mr-H^{\flat}_{\alpha}(X^m|E^m)_{\tilde{\rho}}\right\}.
\end{equation}
Then, for any integer $n$ we construct a hash function $f_n$ as follows. We find the integer $k$ such that $n\in [km, (k+1)m)$ and choose $\mc{Z}_n=\mc{Z}_{m,k}$ as the alphabet set of the extracted randomness and define $f_n$ as
\[
f_n(x_1, \ldots, x_{mk}, x_{mk+1}, \ldots, x_n)=f_{m,k}(x_1, \ldots, x_{mk}).
\]
It is easy to see that
\[
\mc{P}_{f_n}(\rho_{XE}^{\ox n})=\mc{P}_{f_{m,k}}(\rho_{XE}^{\ox mk}) \ox \rho_E^{\ox (n-mk)}.
\]

The performance of $f_n$ can be estimated via
\begin{align}
& F^2\Big(\mc{P}_{f_{m,k}}\big((\tilde{\rho}_{X^mE^m})^{\ox k}\big), \frac{\1_{Z_{m,k}}}{|\mc{Z}_{m,k}|} \ox \sigma_{E^{mk}}\Big) \nb\\
\leq & F^2\Big(\mc{P}_{f_{m,k}}\big((\tilde{\rho}_{X^mE^m})^{\ox k}\big), \frac{\1_{Z_{m,k}}}{|\mc{Z}_{m,k}|}\ox \mc{E}_{\sigma^u_{E^m}}^{\ox k}(\sigma_{E^{mk}})\Big) \nb\\
\leq &  (v_{m,d_E})^k F^2\Big(\mc{P}_{f_{m,k}}(\rho_{XE}^{\ox mk}), \frac{\1_{Z_{m,k}}}{|\mc{Z}_{m,k}|} \ox \mc{E}_{\sigma^u_{E^m}}^{\ox k}(\sigma_{E^{mk}})\Big) \nb\\
= & (v_{m,d_E})^k F^2\Big(\mc{P}_{f_{m,k}}(\rho_{XE}^{\ox mk})\ox \rho_E^{\ox (n-mk)}, \frac{\1_{Z_{m,k}}}{|\mc{Z}_{m,k}|} \ox \mc{E}_{\sigma^u_{E^m}}^{\ox k}(\sigma_{E^{mk}}) \ox \rho_E^{\ox (n-mk)}\Big) \nb\\
\leq & (v_{m,d_E})^k \max_{\omega_{E^n} \in \mc{S}(E^n)} F^2\Big(\mc{P}_{f_n}(\rho_{XE}^{\ox n}), \frac{\1_{Z_{n}}}{|\mc{Z}_{n}|} \ox \omega_{E^n}\Big), \label{eq:main}
\end{align}
where the second line comes from the data processing inequality for the fidelity and the fact that $\mc{P}_{f_{m,k}}\big((\tilde{\rho}_{X^mE^m})^{\ox k}\big)$ is invariant under the action of $\mc{E}_{\sigma^u_{E^m}}^{\ox k}$, and the third line is due to Lemma~\ref{lem:appen1}.
By Eq.~(\ref{eq:rate}), we have
\beq
\liminf_{n\rar\infty} \frac{1}{n} \log|\mc{Z}_n|=\frac{1}{m}\liminf_{k\rar\infty} \frac{1}{k} \log|Z_{m,k}| \geq r.
\eeq
Hence, Eq.~(\ref{eq:intermit}), Eq.~(\ref{eq:main}) and the definition of $E_{\rm{sc}}^{\rm{pa}}(\rho_{XE}, r)$ let us arrive at
\begin{align}
E_{\rm{sc}}^{\rm{pa}}(\rho_{XE}, r)
\leq &\limsup_{n\rar\infty} \frac{-1}{n}\log\max_{\omega_{E^n} \in \mc{S}(E^n)} F^2\big(\mc{P}_{f_n}(\rho_{XE}^{\ox n}), \frac{\1_{Z_n}}{|\mc{Z}_n|} \ox \omega_{E^n}\big) \nb\\
\leq & \sup_{\frac{1}{2} \leq \alpha \leq 1} \frac{1-\alpha}{\alpha}\left\{r-\frac{H^{\flat}_{\alpha}(X^m|E^m)_{\tilde{\rho}}}{m}\right\}+\frac{\log v_{m,d_E}}{m}. \label{eq:bepinching}
\end{align}

Now, we proceed to bound $E_{\rm{sc}}^{\rm{pa}}(\rho_{XE}, r)$ based on Eq.~(\ref{eq:bepinching}). We have
\begin{align}
E_{\rm{sc}}^{\rm{pa}}(\rho_{XE}, r) &\leq \sup_{\frac{1}{2} \leq \alpha \leq 1} \frac{1-\alpha}{\alpha}\left\{\inf_{\sigma_{E^m} \in \mc{S}(E^m)}\frac{D_{\alpha}^{\flat}\big(\mc{E}_{\sigma^u_{E^m}}(\rho_{XE}^{\ox m})\big\|\1_{X}^{\ox m} \ox \sigma_{E^m}\big)}{m}+r\right\}+\frac{\log v_{m,d_E}}{m} \nb\\
&\leq  \sup_{\frac{1}{2} \leq \alpha \leq 1} \frac{1-\alpha}{\alpha}\left\{\frac{D_{\alpha}^{\flat}\big(\mc{E}_{\sigma^u_{E^m}}(\rho_{XE}^{\ox m})\big\|\1_{X}^{\ox m} \ox \sigma^u_{E^m}\big)}{m}+r\right\}+\frac{\log v_{m,d_E}}{m} \nb\\
&= \sup_{\frac{1}{2} \leq \alpha \leq 1} \frac{1-\alpha}{\alpha}\left\{\frac{D_{\alpha}^{*}\big(\mc{E}_{\sigma^u_{E^m}}(\rho_{XE}^{\ox m})\big\|\1_{X}^{\ox m} \ox \sigma^u_{E^m}\big)}{m}+r\right\}+\frac{\log v_{m,d_E}}{m} \nb\\
&\leq \sup_{\frac{1}{2} \leq \alpha \leq 1} \frac{1-\alpha}{\alpha}\left\{\frac{D_{\alpha}^{*}\big(\rho_{XE}^{\ox m}\big\|\1_{X}^{\ox m} \ox \sigma^u_{E^m}\big)}{m}+r\right\}+\frac{\log v_{m,d_E}}{m} \nb\\
&\leq \sup_{\frac{1}{2} \leq \alpha \leq 1} \frac{1-\alpha}{\alpha}\left\{\inf_{\sigma_E} \frac{D_{\alpha}^{*}\big(\rho_{XE}^{\ox m}\big\|\1_{X}^{\ox m} \ox \sigma_E^{\ox m}\big)}{m}+r+\frac{\log v_{m,d_E}}{m}\right\}+\frac{\log v_{m,d_E}}{m} \nb\\
&= \sup_{\frac{1}{2} \leq \alpha \leq 1} \frac{1-\alpha}{\alpha}\left\{r-H^{*}_{\alpha}(X|E)_{\rho}+\frac{\log v_{m,d_E}}{m}\right\}+\frac{\log v_{m,d_E}}{m} \nb\\
&\leq \sup_{\frac{1}{2} \leq \alpha \leq 1} \frac{1-\alpha}{\alpha}\left\{r-H^{*}_{\alpha}(X|E)_{\rho}\right\}+\frac{2\log v_{m,d_E}}{m}, \label{eq:fina}
\end{align}
where the third line is because $\mc{E}_{\sigma^u_{E^m}}(\rho_{XE}^{\ox m})$ commutes with $\1_{X}^{\ox m} \ox \sigma^u_{E^m}$, the fourth line is due to the data processing inequality for the
sandwiched R\'enyi divergence and for the fifth line we have used Lemma~\ref{lem:sym} and
Proposition~\ref{prop:mainpro}~(\romannumeral2).
Because Eq.~(\ref{eq:fina}) holds for any $m \in \mathbb{N}$, we complete the proof by letting $m \rar \infty$.
\end{proof}

\subsection{Proof of the optimality part}
In this subsection, we prove the optimality part of Theorem~\ref{theorem:privacy}.

\medskip
\begin{proof}
Let $\big\{f_n : \mc{X}^{\times n} \rar \mc{Z}_n \big\}_{n \in \mathbb{N}}$ be an arbitrary sequence
of hash functions such that
\beq\label{eq:precond}
\liminf_{n\rar\infty} \frac{1}{n} \log |\mc{Z}_n| \geq r.
\eeq
Then by Lemma~\ref{lem:LWD} in the Appendix, for any $\frac{1}{2} < \alpha < 1$ and $\beta>1$ such that $\frac{1}{\alpha}+\frac{1}{\beta}=2$, and for any state $\omega_{E^n}\in\mc{S}(E^n)$, we have
\begin{align}
&\frac{2\alpha}{1-\alpha} \log F\big(\mc{P}_{f_n}(\rho_{XE}^{\ox n}), \frac{\1_{Z_n}}{|\mc{Z}_n|} \ox \omega_{E^n}\big) \nb\\
\leq&H^{*}_{\alpha}(Z_n|E^n)_{\mc{P}_{f_n}(\rho_{XE}^{\ox n})}-H^{*}_{\beta}(Z_n|E^n)_{\frac{\1_{Z_n}}{|\mc{Z}_n|} \ox \omega_{E^n}}  \nb\\
=&H^{*}_{\alpha}(Z_n|E^n)_{\mc{P}_{f_n}(\rho_{XE}^{\ox n})}-\log |\mc{Z}_n|. \label{eq:optpa1}
\end{align}
Let $V_{f_n}:\mc{H}_{X^n}\rar\mc{H}_{X^n}\ox\mc{H}_{Z_n}$ be an isometry such that
$V_{f_n}\ket{x^n}=\ket{x^n}\ox\ket{f_n(x^n)}$ for any $x^n\in\mc{X}^{\times n}$. Then
\beq
\mc{P}_{f_n}(\rho_{XE}^{\ox n})=\tr_{X^n}\left[V_{f_n}\rho_{XE}^{\ox n}V_{f_n}^\dg \right].
\eeq
So,
\begin{align}
     H^{*}_{\alpha}(Z_n|E^n)_{\mc{P}_{f_n}(\rho_{XE}^{\ox n})}
\leq&H^{*}_{\alpha}(X^nZ_n|E^n)_{V_{f_n}\rho_{XE}^{\ox n}V_{f_n}^\dg } \nb\\
  = &H^{*}_{\alpha}(X^n|E^n)_{\rho_{XE}^{\ox n}}  \nb\\
  = &n H^{*}_{\alpha}(X|E)_{\rho_{XE}},   \label{eq:optpa11}
\end{align}
where the first line is by Proposition~\ref{prop:mainpro}~(\romannumeral5), and the second line is because a local isometry does not change the sandwiched R\'enyi conditional entropy. Eqs.~\eqref{eq:precond},\eqref{eq:optpa1} and \eqref{eq:optpa11} imply
\begin{align}
&\limsup_{n\rar\infty} \frac{-1}{n}\log \max_{\omega_{E^n} \in \mc{S}(E^n)}F^2\big(\mc{P}_{f_n}(\rho_{XE}^{\ox n}),\frac{\1_{Z_n}}{|\mc{Z}_n|}\ox\omega_{E^n}\big) \nb\\
\geq &\sup_{\frac{1}{2}<\alpha <1}\frac{1-\alpha}{\alpha}\big\{r-H^{*}_{\alpha}(X|E)_{\rho}\big\} \nb\\
=&\sup_{\frac{1}{2} \leq \alpha \leq 1}
  \frac{1-\alpha}{\alpha}\big\{r-H^{*}_{\alpha}(X|E)_{\rho}\big\}. \label{eq:optpa2}
\end{align}
Noticing that Eq.~\eqref{eq:optpa2} holds for any sequence of hash functions $\{f_n\}_n$ that satisfies Eq.~\eqref{eq:precond}, by the definition of $E_{\rm{sc}}^{\rm{pa}}(\rho_{XE}, r)$, we get
\begin{equation}
E_{\rm{sc}}^{\rm{pa}}(\rho_{XE}, r)
\geq\sup_{\frac{1}{2}\leq\alpha\leq 1}\frac{1-\alpha}{\alpha}\left\{r-H^*_{\alpha}(X|E)_\rho\right\}.
\end{equation}
\end{proof}

\begin{remark}
Eq.~\eqref{eq:optpa2} holds as well if we replace the ``$\limsup$'' with a ``$\liminf$'', which makes the statement stronger. Recall that in the proof of the achievability part, we have upper bounded the strong converse exponent defined with the ``$\limsup$'' (cf.\ Eq.~\eqref{eq:bepinching}), which is stronger than that with a ``$\liminf$''. So our proof shows that, in the definition of $E_{\rm{sc}}^{\rm{pa}}(\rho_{XE}, r)$ in Eq.~\eqref{eq:defscpa}, we can use either ``$\limsup$'' or ``$\liminf$'' to quantify the rate of exponential decay of $\mathfrak{P}^{\rm pa}(\rho_{XE}^{\ox n},f_n)$, leading to the same result of the strong converse exponent.
\end{remark}

\section{Proof of the Strong Converse Exponent for Quantum Information Decoupling}
  \label{sec:proof-dqi}
In this section, we prove Theorem~\ref{theorem:maindec}. The proof is divided into the achievability part and the optimality part, with the former accomplished in the first three subsections, and the latter done in the last subsection. It is worth pointing out that, by proving the achievability part using a sequence of standard decoupling schemes and dealing with the optimality part for any catalytic decoupling schemes, we have shown that Theorem~\ref{theorem:maindec} actually holds for both standard decoupling and catalytic decoupling.

\subsection{Variational expression}
At first, we introduce a useful intermediate quantity $L(\rho_{RA}, r)$.
\begin{definition}
Let $\rho_{RA} \in \mc{S}({RA})$ and $r \geq 0$. The quantity $L(\rho_{RA}, r)$ is defined as
\begin{equation}
L(\rho_{RA}, r):=\sup_{\frac{1}{2} \leq \alpha \leq 1} \frac{1-\alpha}{\alpha}\left\{I_{\alpha}^{\flat}(R:A)_\rho-2r\right\}.
\end{equation}
\end{definition}

In the following proposition, we prove a variational expression for $L(\rho_{RA}, r)$.
\begin{proposition}
\label{prop:vardec}
Let $\rho_{RA} \in \mc{S}(RA)$ and $r \geq 0$. We have
\begin{equation}
L(\rho_{RA}, r)=\inf_{\tau_{RA} \in \mc{S}_{\rho}(RA)} \left\{ D(\tau_{RA} \| \rho_{RA})+|I(R:A)_{\tau}-2r|^+\right\}.
\end{equation}
\end{proposition}

\begin{proof}
For $\frac{1}{2} \leq \alpha < 1$, we have
\begin{align}
 & \frac{1-\alpha}{\alpha}\left\{I_{\alpha}^{\flat}(R:A)_\rho-2r\right\} \nb\\
=& \frac{1-\alpha}{\alpha}
   \left\{\inf_{\sigma_R,\sigma_A}D_{\alpha}^{\flat}(\rho_{RA}\|\sigma_{R}\ox\sigma_{A})-2r\right\} \nb\\
=&\frac{1-\alpha}{\alpha}
  \left\{\inf_{\sigma_R, \sigma_A} \inf_{\tau_{RA} \in \mc{S}(RA)}\Big\{D(\tau_{RA} \|\sigma_{R} \ox \sigma_{A} )+\frac{\alpha}{1-\alpha}D(\tau_{RA}\| \rho_{RA}) \Big\}-2r \right\} \nb\\
=&\frac{1-\alpha}{\alpha}
  \left\{\inf_{\tau_{RA} \in \mc{S}(RA)}\Big\{D(\tau_{RA} \|\tau_{R} \ox \tau_{A} )
   +\frac{\alpha}{1-\alpha}D(\tau_{RA}\| \rho_{RA}) \Big\}-2r \right\} \nb\\
=&\inf_{\tau_{RA} \in \mc{S}_\rho(RA)}
  \left\{D(\tau_{RA} \| \rho_{RA})+\frac{1-\alpha}{\alpha}\big(I(R:A)_{\tau}-2r\big)\right\}, \label{eq:v3}
\end{align}
where the second line is by definition, the third line is due to the variational expression for $D_{\alpha}^{\flat}(\rho \| \sigma)$ (see Proposition~\ref{prop:mainpro}~(\romannumeral3)). When $\alpha=1$, we can directly check that the first line and the last line of Eq.~\eqref{eq:v3} are still equal. So,
\begin{align}
  L(\rho_{RA}, r)
=&\sup_{\frac{1}{2} \leq \alpha \leq 1} \inf_{\tau_{RA} \in \mc{S}_\rho(RA)}
  \left\{D(\tau_{RA} \| \rho_{RA})+\frac{1-\alpha}{\alpha}\big(I(R:A)_{\tau}-2r\big)\right\} \nb\\
=&\sup_{\lambda \in [0,1]} \inf_{\tau_{RA} \in \mc{S}_\rho(RA)}
\big\{D(\tau_{RA} \| \rho_{RA})+\lambda(I(R:A)_{\tau}-2r)  \big\} \nb\\
=& \inf_{\tau_{RA} \in \mc{S}_\rho(RA)} \sup_{\lambda \in [0,1]}
\big\{D(\tau_{RA} \| \rho_{RA})+\lambda(I(R:A)_{\tau}-2r)  \big\} \nb\\
=&\inf_{\tau_{RA} \in \mc{S}_\rho(RA)} \big\{ D(\tau_{RA} \| \rho_{RA})+|I(R:A)_\tau-2r|^+ \big\}.
\end{align}
In the above derivation, the interchange of the supremum and the infimum is by Sion's minimax theorem (Lemma~\ref{lem:minimax} in the Appendix). The convexity of the function $\tau_{RA}\mapsto D(\tau_{RA} \| \rho_{RA})+\lambda(I(R:A)_\tau-2r)$ that the Sion's minimax theorem requires is not obvious. To see this, we write
\[
\begin{split}
&D(\tau_{RA} \| \rho_{RA})+\lambda(I(R:A)_\tau-2r)\\
=&-\tr \tau_{RA} \log \rho_{RA} -\lambda H(A|R)_\tau -\lambda H(R|A)_\tau-(1-\lambda)H(RA)_\tau-2\lambda r.
\end{split}
\]
Because $-\tr \tau_{RA} \log \rho_{RA}$, $-\lambda H(A|R)_\tau$, $ -\lambda H(R|A)_\tau$ and
$-(1-\lambda)H(RA)_\tau$ are all convex functions of $\tau_{RA}$ (see, e.g., \cite{Wilde2013quantum} for the latter three), the result follows.
\end{proof}

\subsection{An upper bound in terms of the log-Euclidean R{\'e}nyi mutual information}
Now we show that $L(\rho_{RA}, r)$ is achievable as a strong converse rate for quantum information decoupling.
\begin{theorem}
\label{thm:Fdeco}
Let $\rho_{RA} \in \mc{S}(RA)$ and $r \geq 0$. It holds that
\begin{equation}
E^{\rm{dec}}_{\rm{sc}}(\rho_{RA}, r) \leq L(\rho_{RA}, r).
\end{equation}
Furthermore, this bound holds even if we are restricted to the setting of standard decoupling.
\end{theorem}
In order to prove Theorem~\ref{thm:Fdeco}, we introduce
\begin{align}
L_1(\rho_{RA}, r)&:=\inf_{\tau_{RA} \in \mc{O}_1} D(\tau_{RA} \| \rho_{RA}), \\
L_2(\rho_{RA}, r)&:=\inf_{\tau_{RA} \in \mc{O}_2} \left\{D(\tau_{RA} \| \rho_{RA})+I(R:A)_{\tau}-2r \right\},
\end{align}
with
\begin{align}
\mc{O}_1 &:= \{\tau_{RA}~|~\tau_{RA} \in \mc{S}_\rho(RA),~I(R:A)_\tau < 2r\}, \\
\mc{O}_2 &:= \{\tau_{RA}~|~\tau_{RA} \in \mc{S}_\rho(RA),~I(R:A)_{\tau} \geq 2r\}.
\end{align}
It is obviously seen from Proposition~\ref{prop:vardec} that
\begin{equation}
\label{eq:LL1L2}
L(\rho_{RA}, r)=\min \{L_1(\rho_{RA}, r), L_2(\rho_{RA}, r) \}.
\end{equation}
So, it suffices to show that $E^{\rm{dec}}_{\rm{sc}}(\rho_{RA}, r)$ is upper bounded
by both $L_1(\rho_{RA}, r)$ and $L_2(\rho_{RA}, r)$. We accomplish the proof in
the following Lemma~\ref{lem:decF1} and Lemma~\ref{lem:decouF2}, respectively.

\begin{lemma}
\label{lem:decF1}
For $\rho_{RA} \in \mc{S}(RA)$ and $r \geq 0$, we have
\begin{equation}
E^{\rm{dec}}_{\rm{sc}}(\rho_{RA}, r) \leq L_1(\rho_{RA}, r).
\end{equation}
\end{lemma}
\begin{proof}
The definition of $L_1(\rho_{RA}, r)$ ensures that, for any $\delta>0$,
there exists a state $\tau_{RA} \in \mc{S}_\rho(RA)$ such that
\begin{align}
&D(\tau_{RA} \| \rho_{RA}) \leq L_1(\rho_{RA}, r)+\delta, \label{eq:reupper}\\
&\frac{1}{2}I(R:A)_\tau < r. \label{eq:muI}
\end{align}

With Eq.~\eqref{eq:muI}, the work~\cite{ADHW2009mother} implies that $r$ is an asymptotically achievable rate of decoupling cost for the state $\tau_{RA}$. Moreover, this can be done by employing standard decoupling schemes in which the catalytic systems are of dimension one and can be omitted. Specifically, there exists a sequence of isometries $\{\mc{D}_n\equiv U_n:\mc{H}_{A^n}\rar\mc{H}_{\bar{A}_n\tilde{A}_n}\}_{n \in \mathbb{N}}$ which serve as decoupling schemes, such that
\begin{align}
&\lim_{n \rar \infty} \left\|\tr_{\tilde{A}_n}\big[U_{n}\tau_{RA}^{\ox n}
   U_{n}^\dg \big]-\tau_R^{\ox n} \ox \omega_{\bar{A}_n} \right\|_1=0, \label{eq:decadhw-1} \\
&\lim_{n\rar\infty} \frac{1}{n} \log |\tilde{A}_n|=\frac{1}{2}I(R:A)_\tau+\epsilon \leq r,
 \label{eq:decadhw-2}
\end{align}
where $\omega_{\bar{A}_n}$ is the reduced state of $\tr_{\tilde{A}_n}[U_{n}\tau_{RA}^{\ox n}
U_{n}^\dg ]$ on $\bar{A}_n$, and $\epsilon>0$ is a sufficiently small constant. This statement is implicitly included in~\cite{ADHW2009mother}, and we provide a detailed proof in the Appendix (see Lemma~\ref{lem:decoupling}). We remark that Eq.~\eqref{eq:decadhw-1} is equivalent to
\beq
\mathfrak{P}^{\rm dec}\left(\tau_{RA}^{\ox n},\mc{D}_n\right)
\equiv\max_{\gamma_{R^n},\gamma_{\bar{A}_n}}F^2
\left(\tr_{\tilde{A}_n}[U_{n}\tau_{RA}^{\ox n}U_{n}^\dg ],\gamma_{R^n}\ox\gamma_{\bar{A}_n}\right) \rar 1,
\eeq
which can be verified by employing the Fuchs-van de Graaf inequality~\cite{FuchsVan1999cryptographic} between fidelity and trace distance, as well as the triangle inequality and the contractivity property of the trace distance.

We want a result similar to Eq.~\eqref{eq:decadhw-1}, but with the trace distance replaced by the relative entropy. A weaker form sufficient for our purpose can be obtained by making use of the uniform continuity of quantum conditional entropy~\cite{AlickiFannes2004continuity,Winter2016tight}.
At first, we have
\begin{align}
&D\left(\tr_{\tilde{A}_n}[U_{n}\tau_{RA}^{\ox n}U_{n}^\dg ]
\big\|\tau_R^{\ox n}\ox\omega_{\bar{A}_n}\right) \nb\\
=&H(R^n|\bar{A}_n)_{\tau_R^{\ox n}\ox\omega_{\bar{A}_n}}
 -H(R^n|\bar{A}_n)_{\tr_{\tilde{A}_n}[U_{n}\tau_{RA}^{\ox n}U_{n}^\dg ]} \nb\\
\leq & 4\epsilon_n\log|R|^n-2\epsilon_n\log\epsilon_n-2(1-\epsilon_n)\log(1-\epsilon_n)\label{eq:ceconti},
\end{align}
where $\epsilon_n:=\frac{1}{2}\|\tr_{\tilde{A}_n}\big[U_{n}\tau_{RA}^{\ox n}U_{n}^\dg \big]-\tau_R^{\ox n} \ox \omega_{\bar{A}_n}\|_1$, and the last line is a direct application of the Alicki-Fannes inequality~\cite{AlickiFannes2004continuity}. Then, we combine Eq.~\eqref{eq:decadhw-1} and Eq.~\eqref{eq:ceconti} to obtain
\beq\label{eq:reasydec}
\lim_{n\rar\infty}\frac{1}{n}
D\left(\tr_{\tilde{A}_n}[U_{n}\tau_{RA}^{\ox n}U_{n}^\dg ]
\big\|\tau_R^{\ox n}\ox\omega_{\bar{A}_n}\right)=0.
\eeq

Now, we apply the above decoupling schemes to the states $\{\rho_{RA}^{\ox n}\}_n$ and estimate the performance. Our main method is Lemma~\ref{lem:fidelity-re}, from which we deduce that
\begin{align}
     &-\log \mathfrak{P}^{\rm dec}\left(\rho_{RA}^{\ox n},\mc{D}_n\right) \nb\\
  =  &-\log \max_{\gamma_{R^n},\gamma_{\bar{A}_n}}F^2\left(
       \tr_{\tilde{A}_n}[U_{n}\rho_{RA}^{\ox n}U_{n}^\dg ],\gamma_{R^n}\ox\gamma_{\bar{A}_n}\right) \nb\\
\leq &-\log F^2\left(\tr_{\tilde{A}_n}[U_{n}\rho_{RA}^{\ox n}U_{n}^\dg ],
       \tau_R^{\ox n}\ox\omega_{\bar{A}_n}\right) \nb\\
\leq &D\big(\tr_{\tilde{A}_n}[U_{n}\tau_{RA}^{\ox n}U_{n}^\dg ] \big\|
       \tr_{\tilde{A}_n}[U_{n}\rho_{RA}^{\ox n}U_{n}^\dg ]\big)
     +D\left(\tr_{\tilde{A}_n}[U_{n}\tau_{RA}^{\ox n}U_{n}^\dg ] \big\|
       \tau_R^{\ox n}\ox\omega_{\bar{A}_n}\right) \nb\\
\leq &n D(\tau_{RA} \| \rho_{RA})+D\big(\tr_{\tilde{A}_n}[U_{n}\tau_{RA}^{\ox n}U_{n}^\dg ]\big\|
       \tau_R^{\ox n}\ox\omega_{\bar{A}_n}\big) , \label{eq:decachi-11}
\end{align}
where the equality is by definition (Eq.~\eqref{eq:def-p-dec} specialized to the standard setting of decoupling), and the last inequality is due to the data processing inequality of the relative entropy. Eq.~\eqref{eq:reasydec} and Eq.~\eqref{eq:decachi-11} together result in
\beq\label{eq:decachi-12}
\limsup_{n\rar\infty}\frac{-1}{n}\log\mathfrak{P}^{\rm dec}\left(\rho_{RA}^{\ox n},\mc{D}_n\right)
\leq D(\tau_{RA} \| \rho_{RA}).
\eeq

At last, according to the definition of $E^{\rm{dec}}_{\rm{sc}}(\rho_{RA}, r)$, Eq.~\eqref{eq:decadhw-2} and Eq.~\eqref{eq:decachi-12} imply that
\begin{align}
       E^{\rm{dec}}_{\rm{sc}}(\rho_{RA}, r)
\leq & D(\tau_{RA} \| \rho_{RA}) \nb\\
\leq & L_1(\rho_{RA}, r)+\delta,
\end{align}
where the second line is due to Eq.~\eqref{eq:reupper}. Since $\delta>0$ is arbitrary, we complete the proof by letting $\delta\rar 0$.
\end{proof}

\begin{lemma}
\label{lem:decouF2}
For $\rho_{RA} \in \mc{S}(RA)$ and $r \geq 0$, we have
\begin{equation}
E^{\rm{dec}}_{\rm{sc}}(\rho_{RA}, r) \leq L_2(\rho_{RA}, r).
\end{equation}
\end{lemma}

\begin{proof}
According to the definition of $L_2(\rho_{RA}, r)$, there exists
$\tau_{RA} \in \mc{S}_\rho(RA)$ such that
\begin{align}
&D(\tau_{RA} \| \rho_{RA})+I(R:A)_\tau-2r = L_2(\rho_{RA}, r), \\
&I(R:A)_\tau \geq 2r.
\end{align}
Let $r'=\frac{1}{2} I(R:A)_\tau+\epsilon$, where $\epsilon>0$ is a small constant. From the proof of Lemma~\ref{lem:decF1} (cf.\ Eqs.~\eqref{eq:decadhw-2} and~\eqref{eq:decachi-12}), we know that there exist a sequence of standard decoupling schemes $\{U_n:\mc{H}_{A^n} \rar \mc{H}_{\bar{A}_n\tilde{A}_n}\}_{n \in \mathbb{N}}$ in which the catalytic systems are of dimension one, and a sequence of states $\{(\omega_{R^n},\omega_{\bar{A}_n})\}_{n \in \mathbb{N}}$, such that
\begin{align}
\label{eq:extensionF11}
&\limsup_{n\rar\infty} \frac{-1}{n}\log F^2\left(\tr_{\tilde{A}_n}[U_{n}\rho_{RA}^{\ox n} U_{n}^\dg ], \omega_{R^n} \ox \omega_{\bar{A}_n}\right) \leq D(\tau_{RA} \| \rho_{RA}), \\
&\lim_{n\rar\infty} \frac{1}{n}\log|\tilde{A}_n| = r'. \label{eq:extensionF12}
\end{align}
By Uhlmann's theorem~\cite{Uhlmann1976transition}, there exists an extension $\omega_{R^n\bar{A}_n\tilde{A}_n}$ of $\omega_{R^n} \ox \omega_{\bar{A}_n}$ such that
\begin{equation}\label{eq:extensionF2}
F^2\left(\tr_{\tilde{A}_n}[U_{n}\rho_{RA}^{\ox n} U_{n}^\dg ],\omega_{R^n}\ox\omega_{\bar{A}_n}\right)
=F^2\left(U_{n}\rho_{RA}^{\ox n}U_{n}^\dg ,\omega_{R^n\bar{A}_n\tilde{A}_n} \right).
\end{equation}
We construct for each $n\in\mathbb{N}$ a new decoupling scheme. Let $\tilde{A}_n=\bar{A}'_n\tilde{A}'_n$ be a decomposition such that $|\tilde{A}'_n|=\lceil2^{nr}\rceil$. We point out that, to ensure that the dimension of $\bar{A}'_n$ is also an integer, we actually choose $|\bar{A}'_n|=\lceil|\tilde{A}_n|
/|\tilde{A}'_n|\rceil$ and this is an embedding. The new decoupling scheme is to discard the system $\tilde{A}'_n$ instead of $\tilde{A}_n$, and keep the systems $\bar{A}_n$ and $\bar{A}'_n$. The rate of the cost of this sequence of decoupling schemes is
\beq\label{eq:ratecdec}
\lim_{n\rar\infty}\frac{1}{n}\log|\tilde{A}'_n|=r,
\eeq
and the performance is estimated as
\begin{align}
&\max_{\gamma_{R^n},\gamma_{\bar{A}_n\bar{A}'_n}} F^2\left(\tr_{\tilde{A}'_n}[U_{n}\rho_{RA}^{\ox n} U_{n}^\dg ], \gamma_{R^n} \ox \gamma_{\bar{A}_n\bar{A}'_n} \right) \nb\\
\geq &F^2\Big(\tr_{\tilde{A}'_n}[U_{n}\rho_{RA}^{\ox n} U_{n}^\dg ],\omega_{R^n} \ox \omega_{\bar{A}_n} \ox \frac{\1_{\bar{A}'_n}}{|\bar{A}'_n|}\Big)  \nb\\
\geq & \frac{1}{|\bar{A}'_n|^2} F^2\left(\tr_{\tilde{A}'_n}[U_{n}\rho_{RA}^{\ox n} U_{n}^\dg ],\tr_{\tilde{A}'_n}[\omega_{R^n\bar{A}_n\tilde{A}_n}]\right) \nb\\
\geq &\frac{|\tilde{A}'_n|^2}{|\tilde{A}_n|^2}F^2\left(U_{n}\rho_{RA}^{\ox n}U_{n}^\dg ,\omega_{R^n\bar{A}_n\tilde{A}_n} \right). \label{eq:extensionF3}
\end{align}
In Eq.~\eqref{eq:extensionF3}, the last line is by the data processing inequality for the fidelity function, and for the third line we have used the inequality
\begin{equation}
\tr_{\tilde{A}'_n}[\omega_{R^n\bar{A}_n\tilde{A}_n}]=:\omega_{R^n\bar{A}_n\bar{A}'_n} \leq |\bar{A}'_n| \omega_{R^n\bar{A}_n} \ox \1_{\bar{A}'_n}
=|\bar{A}'_n|^2\omega_{R^n} \ox \omega_{\bar{A}_n} \ox \frac{\1_{\bar{A}'_n}}{|\bar{A}'_n|}.
\end{equation}
Making use of Eqs.~(\ref{eq:extensionF11})--(\ref{eq:ratecdec}), we translate Eq.~\eqref{eq:extensionF3} into
\begin{align}
&\limsup_{n\rar\infty} \frac{-1}{n}\log \max_{\gamma_{R^n},\gamma_{\bar{A}_n\bar{A}'_n}} F^2\left(\tr_{\tilde{A}'_n}[U_{n}\rho_{RA}^{\ox n} U_{n}^\dg ],
\gamma_{R^n} \ox \gamma_{\bar{A}_n\bar{A}'_n} \right) \nb\\
\leq &D(\tau_{RA} \| \rho_{RA}) +2(r'-r) \nb\\
  =  &L_2(\rho_{RA}, r)+ 2\epsilon. \label{eq:extensionF4}
\end{align}
At last, according to the definition of $E^{\rm{dec}}_{\rm{sc}}(\rho_{RA}, r)$, Eq.~\eqref{eq:ratecdec} and Eq.~\eqref{eq:extensionF4} result in
\beq
E^{\rm{dec}}_{\rm{sc}}(\rho_{RA}, r)\leq L_2(\rho_{RA}, r)+ 2\epsilon.
\eeq
Because $\epsilon>0$ is arbitrary, letting $\epsilon\rar 0$ we conclude the proof.
\end{proof}

\subsection{Proof of the achievability part: last step}
With Theorem~\ref{thm:Fdeco}, we are now ready to complete the proof of the achievability part of
Theorem~\ref{theorem:maindec}.

\medskip
\begin{proof}
Fix an arbitrary $\delta>0$. Let $N_{\delta}$ be a sufficiently large integer such that $\frac{\log v_{m,d_A}}{m} \leq \delta$ for all $m \geq  N_{\delta}$. Fix such an integer $m \geq  N_{\delta}$ and let $\tilde{\rho}_{R^mA^m}=\mc{E}_{\sigma^u_{R^m}} \ox \mc{E}_{\sigma^u_{A^m}}(\rho_{RA}^{\ox m})$, where $\sigma^u_{R^m}$ and $\sigma^u_{A^m}$ are the particular symmetric states of Lemma~\ref{lem:sym}. Theorem~\ref{thm:Fdeco} tells us that there exist a sequence of isometries
\beq\left\{U_k:\mc{H}_{A^{mk}}\rar\mc{H}_{\bar{A}_{m,k}\tilde{A}_{m,k}}\right\}_{k\in\mathbb{N}}
\eeq
which serve as standard decoupling schemes, and a sequence of states
$\{(\omega_{R^{mk}},\omega_{\bar{A}_{m,k}})\}_{k \in \mathbb{N}}$,
such that
\begin{equation}
\limsup_{k\rar\infty} \frac{1}{k} \log |\tilde{A}_{m,k}| \leq m(r-\delta)
\end{equation}
and
\begin{align}
&\limsup_{k\rar\infty} \frac{-1}{k} \log  F^2\left(\tr_{\tilde{A}_{m,k}}\left[U_{k}\tilde{\rho}_{R^mA^m}^{\ox k} U_{k}^\dg \right],
\omega_{R^{mk}} \ox \omega_{\bar{A}_{m,k}} \right) \nb\\
\leq & \sup_{\frac{1}{2} \leq \alpha \leq 1} \frac{1-\alpha}{\alpha}\left\{I_{\alpha}^{\flat}(R^m:A^m)_{\tilde{\rho}}-2m(r-\delta)\right\}.\label{eq:dec1}
\end{align}

Based on the above result, for any integer $n$, we construct a standard decoupling scheme for $\rho_{RA}^{\ox n}$. By the Stinespring dilation theorem, there is a system $\hat{A}$ of dimension $|\hat{A}|\leq v_{m,d_A}$ and a isometry $V:\mc{H}_{A^m}\rar\mc{H}_{A^m}\ox\mc{H}_{\hat{A}}$, such that
\beq
\mc{E}_{\sigma^u_{A^m}}(\cdot)=
\tr_{\hat{A}}\left[V(\cdot) V^\dg \right].
\eeq
Write $n=mk+l$ with $0\leq l <m$, and $\rho_{RA}^{\ox n}=(\rho_{RA}^{\ox m})^{\ox k}\ox\rho_{RA}^{\ox l}$.
Then the decoupling scheme is constructed as
\beq
\mc{D}_n\equiv U_{k}V^{\ox k}\ox\1_{A^l},
\eeq
where $U_{k}V^{\ox k}$ acts on the first $mk$ copies of the $A$ system, and $\1_{A^l}$ is the identity matrix on the last $l$ copies of the $A$ system. We will keep the system $\bar{A}_{m,k}$ and discard the systems $\tilde{A}_{m,k}$, $\hat{A}^k$, and $A^l$. The resulting state can be written as
\begin{align}
\rho_{R^n\bar{A}_{m,k}}
=&\tr_{\tilde{A}_{m,k} \hat{A}^k A^{l}}
 \left[U_k\left(V(\rho_{RA}^{\ox m})V^\dg \right)^{\ox k} U_k^\dg  \ox \rho_{RA}^{\ox l} \right] \nb\\
=&\tr_{\tilde{A}_{m,k}}\left[U_{k}\big(\mc{E}_{\sigma^u_{A^m}}(\rho_{RA}^{\ox m})\big)^{\ox k}
U_{k}^\dg \right]\ox \rho_R^{\ox l}. \label{eq:remainingst}
\end{align}

Now, we evaluate the performance of this scheme.
\begin{align}
&\max_{\gamma_{R^{n}}, \gamma_{\bar{A}_{m,k}}}
 F^2\left(\rho_{R^n\bar{A}_{m,k}},\gamma_{R^{n}}\ox \gamma_{\bar{A}_{m,k}}\right) \nb\\
\geq &F^2\left(\rho_{R^n\bar{A}_{m,k}},\mc{E}_{\sigma^u_{R^m}}^{\ox k}(\omega_{R^{mk}})
      \ox \rho_{R}^{\ox l} \ox \omega_{\bar{A}_{m,k}}\right) \nb\\
=&F^2\left(\tr_{\tilde{A}_{m,k}}
     \left[U_{k}\big(\mc{E}_{\sigma^u_{A^m}}(\rho_{RA}^{\ox m})\big)^{\ox k}U_{k}^\dg \right],
     \mc{E}_{\sigma^u_{R^m}}^{\ox k}(\omega_{R^{mk}}) \ox \omega_{\bar{A}_{m,k}} \right) \nb\\
\geq & (v_{m,d_R})^{-k} F^2\left(\tr_{\tilde{A}_{m,k}}\left[U_{k}\tilde{\rho}_{R^mA^m}^{\ox k}
      U_k^\dg \right],\mc{E}_{\sigma^u_{R^m}}^{\ox k}(\omega_{R^{mk}})\ox\omega_{\bar{A}_{m,k}}\right) \nb\\
\geq & (v_{m,d_R})^{-k} F^2\left(\tr_{\tilde{A}_{m,k}}\left[U_{k}\tilde{\rho}_{R^mA^m}^{\ox k}
      U_k^\dg \right],\omega_{R^{mk}}\ox\omega_{\bar{A}_{m,k}}\right) \label{eq:dec2},
\end{align}
where for the third line we used Eq.~\eqref{eq:remainingst}, the fourth line is by  Lemma~\ref{lem:appen1} and the last line is due to the data processing inequality for the fidelity function. The asymptotic rate of cost of this sequence of decoupling schemes can be bounded as
\begin{equation}
\label{eq:asyrate}
\limsup_{n\rar\infty} \frac{1}{n} \log |\tilde{A}_{m,k}||\hat{A}|^k|A|^{l} \leq r-\delta+\frac{\log v_{m,d_A}}{m} \leq r.
\end{equation}
Hence, Eq.~(\ref{eq:dec1}), Eq.~(\ref{eq:dec2}) and Eq.~(\ref{eq:asyrate}) together yield
\begin{equation}
\label{eq:decinter}
E_{\rm{sc}}^{\rm{dec}}(\rho_{RA}, r) \leq \sup_{\frac{1}{2} \leq \alpha \leq 1} \frac{1-\alpha}{\alpha}\left\{\frac{I_{\alpha}^{\flat}(R^m:A^m)_{\tilde{\rho}}}{m}-2r+2\delta\right\} +\frac{\log v_{m,d_R}}{m}.
\end{equation}

To complete the proof, we further bound Eq.~(\ref{eq:decinter}) as follows.
\begin{align}
&E_{\rm{sc}}^{\rm{dec}}(\rho_{RA}, r) \nb\\
\leq &\sup_{\frac{1}{2} \leq \alpha \leq 1} \frac{1-\alpha}{\alpha} \left\{\inf_{\sigma_{R^m} , \sigma_{A^m}} \frac{D_{\alpha}^{\flat}(\mc{E}_{\sigma^u_{R^m}} \ox \mc{E}_{\sigma^u_{A^m}} (\rho_{RA}^{\ox m}) \| \sigma_{R^m} \ox \sigma_{A^m} )}{m} -2r \right\} +\frac{\log v_{m,d_R}}{m}+2\delta \nb\\
\leq &\sup_{\frac{1}{2} \leq \alpha \leq 1} \frac{1-\alpha}{\alpha} \left\{ \frac{D_{\alpha}^{\flat}(\mc{E}_{\sigma^u_{R^m}} \ox \mc{E}_{\sigma^u_{A^m}}(\rho_{RA}^{\ox m}) \| \sigma^u_{R^m} \ox \sigma^u_{A^m}  )}{m}-2r \right\} +\frac{\log v_{m,d_R}}{m}+2\delta \nb\\
=&\sup_{\frac{1}{2} \leq \alpha \leq 1} \frac{1-\alpha}{\alpha} \left\{ \frac{D_{\alpha}^{*}(\mc{E}_{\sigma^u_{R^m}} \ox \mc{E}_{\sigma^u_{A^m}}(\rho_{RA}^{\ox m}) \| \sigma^u_{R^m} \ox \sigma^u_{A^m} )}{m}-2r \right\} +\frac{\log v_{m,d_R}}{m}+2\delta \nb\\
\leq &\sup_{\frac{1}{2} \leq \alpha \leq 1} \frac{1-\alpha}{\alpha} \left\{ \frac{D_{\alpha}^{*}(\rho_{RA}^{\ox m} \| \sigma^u_{R^m} \ox \sigma^u_{A^m} )}{m}-2r \right\} +\frac{\log v_{m,d_R}}{m}+2\delta \nb\\
\leq &\sup_{\frac{1}{2} \leq \alpha \leq 1} \frac{1-\alpha}{\alpha} \left\{ \inf_{\sigma_R, \sigma_A} \frac{D_{\alpha}^{*}(\rho_{RA}^{\ox m} \| \sigma_R^{\ox m} \ox \sigma_{A}^{\ox m})}{m}-2r+ \frac{\log v_{m,d_A}}{m}+\frac{\log v_{m,d_R}}{m} \right\} +\frac{\log v_{m,d_R}}{m}+2\delta \nb\\
\leq &\sup_{\frac{1}{2} \leq \alpha \leq 1} \frac{1-\alpha}{\alpha} \left\{ I_{\alpha}^*(R:A)_\rho-2r \right\} +\frac{2\log v_{m,d_R}}{m}+\frac{\log v_{m,d_A}}{m}+2\delta, \label{eq:declast}
\end{align}
where the fourth line is because $\mc{E}_{\sigma^u_{R^m}} \ox \mc{E}_{\sigma^u_{A^m}} (\rho_{RA}^{\ox m})$ commutes with $\sigma^u_{R^m} \ox \sigma^u_{A^m} $, the fifth line is due to the data processing inequality for the sandwiched R\'enyi divergence, the sixth line is by Lemma~\ref{lem:sym} and Proposition~\ref{prop:mainpro}~(\romannumeral2). Noticing that Eq.~(\ref{eq:declast}) holds for any $\delta>0$ and any $m>N_\delta$, by first letting $m \rar \infty$ and then letting $\delta\rar 0$ we obtain
\begin{equation}
\label{eq:lastdec1}
E_{\rm{sc}}^{\rm{dec}}(\rho_{RA}, r) \leq \sup_{\frac{1}{2} \leq \alpha \leq 1} \frac{1-\alpha}{\alpha}
\left\{ I_{\alpha}^*(R:A)_\rho-2r \right\}.
\end{equation}
It is easy to prove by definition that, for any integer $m$,
\begin{equation}
\label{eq:lasetdec2}
E_{\rm{sc}}^{\rm{dec}}(\rho_{RA}, r)=\frac{1}{m}E_{\rm{sc}}^{\rm{dec}}(\rho^{\ox m}_{RA}, mr).
\end{equation}
Applying Eq.~(\ref{eq:lastdec1}) to bound the right hand side of Eq.~(\ref{eq:lasetdec2}) and then taking the limit $m\rar\infty$, we get
\begin{equation}\label{eq:lasetdec3}
E_{\rm{sc}}^{\rm{dec}}(\rho_{RA}, r)
\leq \lim_{m\rar\infty}\sup_{\frac{1}{2} \leq \alpha \leq 1} \frac{1-\alpha}{\alpha}
\left\{\frac{1}{m}I_{\alpha}^*(R^m:A^m)_{\rho^{\ox m}}-2r\right\}.
\end{equation}
At last, we show in Lemma~\ref{lem:swap} that the limit in Eq.~\eqref{eq:lasetdec3} can be taken before the supremum. This lets us arrive at
\begin{equation}
E_{\rm{sc}}^{\rm{dec}}(\rho_{RA}, r)
\leq \sup_{\frac{1}{2} \leq \alpha \leq 1} \frac{1-\alpha}{\alpha}
\left\{I_{\alpha}^{*, \rm{reg}}(R:A)_\rho-2r\right\},
\end{equation}
completing the proof.
\end{proof}

\medskip
\begin{lemma}
\label{lem:swap}
Let $\rho_{RA} \in \mc{S}(RA)$ and $r \geq 0$. It holds that
\begin{equation}\label{eq:swap}
\lim_{m\rar\infty}\sup_{\frac{1}{2} \leq \alpha \leq 1} \frac{1-\alpha}{\alpha}
\left\{\frac{1}{m}I_{\alpha}^*(R^m:A^m)_{\rho^{\ox m}}-2r\right\}
=\sup_{\frac{1}{2} \leq \alpha \leq 1} \frac{1-\alpha}{\alpha}
\left\{I_{\alpha}^{*, \rm{reg}}(R:A)_\rho-2r\right\}.
\end{equation}
\end{lemma}
\begin{proof}
Two simple properties of the sandwiched R\'enyi mutual information $I_{\alpha}^*$ are the reason behind Eq.~\eqref{eq:swap}. First, for any $m\in\mathbb{N}$, the function $\alpha\mapsto I_{\alpha}^*(R^m:A^m)
_{\rho^{\ox m}}$ is upper semicontinuous on $[\frac{1}{2},1]$. Second, for any $\alpha\in[\frac{1}{2},1]$, the sequence $\{I_{\alpha}^*(R^m:A^m)_{\rho^{\ox m}}\}_m$ is subadditive. Note that a sequence $\{a_m\}_{m\in\mathbb{N}}$ is subadditive if $a_{m+n}\leq a_m + a_n$ for all $m,n\in\mathbb{N}$. The latter property follows directly from the definition of $I_{\alpha}^*$ (cf.\ Eq.~\eqref{eq:srmi-de}). We now show the former. Let $\frac{1}{2}\leq\alpha\leq\alpha'\leq 1$ be arbitrary, and let $\bar{\sigma}_R(\alpha)$ and $\bar{\omega}_A(\alpha)$ be the maximizers in the definition of $I_{\alpha}^*(R:A)_{\rho}$ such that $I_{\alpha}^*(R:A)_{\rho}=D_{\alpha}^*(\rho_{RA}\|\bar{\sigma}_R(\alpha)\ox\bar{\omega}_A(\alpha))$. By definition and the monotonicity of $D_{\alpha}^*$ with respect to the R\'enyi parameter (Proposition~\ref{prop:mainpro}~(\romannumeral1)), we have
\begin{align}
     &D_{\alpha'}^*(\rho_{RA}\|\bar{\sigma}_R(\alpha)\ox\bar{\omega}_A(\alpha))   \nb \\
\geq &D_{\alpha'}^*(\rho_{RA}\|\bar{\sigma}_R(\alpha')\ox\bar{\omega}_A(\alpha')) \nb \\
\geq &D_{\alpha}^*(\rho_{RA}\|\bar{\sigma}_R(\alpha')\ox\bar{\omega}_A(\alpha'))  \nb \\
\geq &D_{\alpha}^*(\rho_{RA}\|\bar{\sigma}_R(\alpha)\ox\bar{\omega}_A(\alpha)).   \label{eq:swap-1}
\end{align}
Fixing $\alpha$ and letting $\alpha'\searrow\alpha$, we get from Eq.~\eqref{eq:swap-1} that
\begin{align}
&~~~~I_{\alpha'}^*(R:A)_{\rho}-I_{\alpha}^*(R:A)_{\rho} \nb\\
&\leq D_{\alpha'}^*(\rho_{RA}\|\bar{\sigma}_R(\alpha)\ox\bar{\omega}_A(\alpha))
     -D_{\alpha}^*(\rho_{RA}\|\bar{\sigma}_R(\alpha)\ox\bar{\omega}_A(\alpha))  \nb\\
&\rar 0. \label{eq:swap-2-1}
\end{align}
Fixing $\alpha'$, for any $\alpha\leq\alpha'$ we get from Eq.~\eqref{eq:swap-1} that
\begin{align}
I_{\alpha}^*(R:A)_{\rho}-I_{\alpha'}^*(R:A)_{\rho} \leq 0. \label{eq:swap-2-2}
\end{align}
Eq.~\eqref{eq:swap-2-1} and Eq.~\eqref{eq:swap-2-2} together show that the function $\alpha\mapsto I_{\alpha}^*(R:A)_{\rho}$ is upper semicontinuous on $[\frac{1}{2},1]$. The upper semi-continuity of $\alpha\mapsto I_{\alpha}^*(R^m:A^m)_{\rho^{\ox m}}$ for $m\geq 2$ follows by the same argument.

We set
\beq
f_m(\alpha):=\frac{1-\alpha}{\alpha}\left\{\frac{1}{m}I_{\alpha}^*(R^m:A^m)_{\rho^{\ox m}}-2r\right\}
\eeq
for brevity and rewrite Eq.~\eqref{eq:swap} as
\begin{equation}\label{eq:swap-3}
 \lim_{m\rar\infty}\sup_{\frac{1}{2} \leq \alpha \leq 1} f_m(\alpha)
=\sup_{\frac{1}{2} \leq \alpha \leq 1}\lim_{m\rar\infty} f_m(\alpha).
\end{equation}
The subadditivity of $\{I_{\alpha}^*(R^m:A^m)_{\rho^{\ox m}}\}_m$ implies that for any $\alpha\in
[\frac{1}{2},1]$, the sequence $\{mf_m(\alpha)\}_m$ is subadditive, which in turn implies that the sequence $\{\sup_{\frac{1}{2} \leq \alpha \leq 1} mf_m(\alpha)\}_m$ is also subadditive. The subadditivity of these sequences, together with the obvious fact that they are bounded from below, further implies that the limits in Eq.~\eqref{eq:swap-3} do exist. Now, we can easily show the ``$\geq$'' part of Eq.~\eqref{eq:swap-3}. Since $\sup_{\frac{1}{2} \leq \alpha \leq 1} f_m(\alpha)\geq f_m(\alpha')$ holds for any $m\in\mathbb{N}$ and $\alpha'\in[\frac{1}{2},1]$, we take limit with $m\rar\infty$ at both sides and then take supremum over $\alpha'\in[\frac{1}{2},1]$ on the right hand side to get
\begin{equation}\label{eq:swap-4}
    \lim_{m\rar\infty}\sup_{\frac{1}{2} \leq \alpha \leq 1} f_m(\alpha)
\geq\sup_{\frac{1}{2} \leq \alpha \leq 1}\lim_{m\rar\infty} f_m(\alpha).
\end{equation}
Next, we prove the ``$\leq$'' part of Eq.~\eqref{eq:swap-3}. Set $\bar{f}_m(\alpha):=f_{2^m}(\alpha)$. It suffices to show
\begin{equation}\label{eq:swap-5}
    \lim_{m\rar\infty}\sup_{\frac{1}{2} \leq \alpha \leq 1} \bar{f}_m(\alpha)
\leq\sup_{\frac{1}{2} \leq \alpha \leq 1}\lim_{m\rar\infty} \bar{f}_m(\alpha),
\end{equation}
because a convergent sequence has the same limit as any of its subsequence. Due to the subadditivity of $\{mf_m(\alpha)\}_m$, we can check that for any $\alpha$ the sequence $\{\bar{f}_m(\alpha)\}_m$ is monotonically non-increasing. For the same reason, the sequence $\{\sup_{\frac{1}{2} \leq \alpha \leq 1} \bar{f}_m(\alpha)\}_m$ is monotonically non-increasing too. Fixing an arbitrary $\epsilon>0$, we set \[L(\epsilon):=\lim\limits_{m\rar\infty}\sup\limits_{\frac{1}{2} \leq \alpha \leq 1} \bar{f}_m(\alpha)-\epsilon\]
and define
\beq
\mc{A}_m:=\left\{\alpha~\big|~\frac{1}{2}\leq\alpha\leq 1, \bar{f}_m(\alpha)\geq L(\epsilon)\right\}.
\eeq
Then the sets $\mc{A}_m,m\in\mathbb{N}$ satisfy the following properties:
\begin{enumerate}[(i)]
  \item $\mc{A}_1\supseteq\mc{A}_2\supseteq\mc{A}_3\supseteq\cdot\cdot\cdot$\ ,
  \item for any $m\in\mathbb{N}$, $\mc{A}_m$ is nonempty, and
  \item for any $m\in\mathbb{N}$, $\mc{A}_m$ is closed.
\end{enumerate}
The first property is due to the monotonicity of the sequence of functions $\{\bar{f}_m(\alpha)\}_m$. The second property follows from the monotonicity of $\{\sup_{\frac{1}{2} \leq \alpha \leq 1} \bar{f}_m(\alpha)\}_m$ and the definition of $L(\epsilon)$. The third property is because of the upper semi-continuity of $\alpha\mapsto\bar{f}_m(\alpha)$ on $[\frac{1}{2},1]$, which results from the upper semi-continuity of $\alpha\mapsto I_{\alpha}^*(R^m:A^m)_{\rho^{\ox m}}$. According to Cantor's intersection theorem, the three properties together ensure that the set $\bigcap_{m\in\mathbb{N}}\mc{A}_m$ is nonempty. So, there exists $\alpha_0\in[\frac{1}{2},1]$ such that $\bar{f}_m(\alpha_0)\geq L(\epsilon)$ for all $m\in\mathbb{N}$, and hence,
\begin{equation}\label{eq:swap-6}
L(\epsilon)\equiv\lim_{m\rar\infty}\sup_{\frac{1}{2} \leq \alpha \leq 1} \bar{f}_m(\alpha)-\epsilon
\leq\sup_{\frac{1}{2} \leq \alpha \leq 1}\lim_{m\rar\infty} \bar{f}_m(\alpha).
\end{equation}
Eventually, since $\epsilon>0$ is chosen arbitrarily, letting $\epsilon\rar0$ in Eq.~\eqref{eq:swap-6} leads to Eq.~\eqref{eq:swap-5}. The combination of Eq.~\eqref{eq:swap-4} and Eq.~\eqref{eq:swap-5} completes the proof of Eq.~\eqref{eq:swap-3}, and we are done.
\end{proof}

\subsection{Proof of the optimality part}
Now, we turn to the proof of the optimality part of Theorem~\ref{theorem:maindec}.

\medskip
\begin{proof}
Let $\left\{\mc{D}_n=(\sigma_{A'_n},\ U_n: \mc{H}_{A^nA'_n} \rar \mc{H}_{\bar{A}_n\tilde{A}_n}) \right\}_{n \in \mathbb{N}}$ be an arbitrary sequence of catalytic decoupling schemes which satisfies
\begin{equation}
\label{eq:opt}
\limsup_{n\rar\infty}\frac{1}{n} \log|\tilde{A}_n| \leq r.
\end{equation}
Let $\{(\omega_{R^n},\omega_{\bar{A}_n})\}_{n \in \mathbb{N}}$ be an arbitrary sequence of quantum states. According to Uhlmann's theorem~\cite{Uhlmann1976transition}, for any $n$ there exists an extension state $\omega_{R^n\bar{A}_n\tilde{A}_n}$ of $\omega_{R^n} \ox \omega_{\bar{A}_n} $ such that
\begin{equation}
\label{eq:opt1}
F\left(\tr_{\tilde{A}_n} \left[U_{n}(\rho_{RA}^{\ox n} \ox \sigma_{A'_n})U_{n}^\dg \right],
\omega_{R^n} \ox \omega_{\bar{A}_n}\right)
=F\left(U_{n}(\rho_{RA}^{\ox n} \ox \sigma_{A'_n})U_{n}^\dg , \omega_{R^n\bar{A}_n\tilde{A}_n}\right).
\end{equation}
Now, fix an $\alpha\in[\frac{1}{2}, 1)$ and let $\beta$ be such that $\frac{1}{\alpha}+\frac{1}{\beta}=2$. Also let $\tau_{\bar{A}_n\tilde{A}_n}\in\mc{S}(\bar{A}_n\tilde{A}_n)$ be the state satisfying
\beq
\bar{I}^{*}_{\beta}(R^n:\bar{A}_n\tilde{A}_n)_\omega=D_{\beta}^*(\omega_{R^n\bar{A}_n\tilde{A}_n} \| \omega_{R^n} \ox \tau_{\bar{A}_n\tilde{A}_n}).
\eeq
For $0<\epsilon<1$, we set $\tau_n(\epsilon)=((1-\epsilon)\omega_{R^n}+\epsilon\pi_{R^n}) \ox ((1-\epsilon) \tau_{\bar{A}_n\tilde{A}_n}+\epsilon \pi_{\bar{A}_n\tilde{A}_n})$, with $\pi_{R^n}$ and $\pi_{\bar{A}_n\tilde{A}_n}$ being the maximally mixed states. Then we have
\begin{align}
&\frac{2\alpha}{1-\alpha} \log F\left(U_{n}(\rho_{RA}^{\ox n} \ox \sigma_{A'_n})U_{n}^\dg , \omega_{R^n\bar{A}_n\tilde{A}_n}\right) \nb \\
\leq &D_{\beta}^*( \omega_{R^n\bar{A}_n\tilde{A}_n}\|\tau_n(\epsilon))
     -D_{\alpha}^*(U_{n}(\rho_{RA}^{\ox n}\ox\sigma_{A'_n})U_{n}^\dg \|\tau_n(\epsilon)) \nb \\
\leq &D_{\beta}^*( \omega_{R^n\bar{A}_n\tilde{A}_n}\|\tau_n(\epsilon))
     -I_{\alpha}^*(R^n:\bar{A}_n\tilde{A}_n)_{U_n(\rho_{RA}^{\ox n} \ox \sigma_{A'_n})U_n^\dg } \nb \\
  =  &D_{\beta}^*( \omega_{R^n\bar{A}_n\tilde{A}_n}\|\tau_n(\epsilon))
     -I_{\alpha}^*(R^n:A^n)_{\rho^{\ox n}},\label{eq:opt2}
\end{align}
where the second line follows from Lemma~\ref{lem:LWDg} in the Appendix (note that Lemma~\ref{lem:LWDg} holds as well for $\alpha=\frac{1}{2}$ and $\beta=\infty$ by continuity), the third line is by the definition of $I^*_\alpha$, and for the last line, we can easily check using the data processing inequality for the sandwiched R\'enyi divergence (Proposition~\ref{prop:mainpro}~(\romannumeral4)), that the R\'enyi mutual information $I_{\alpha}^*$ does not change under locally attaching an independent state and then locally applying a unitary operation. Combining Eq.~\eqref{eq:opt1} and Eq.~(\ref{eq:opt2}), and letting $\epsilon \rar 0$, we get
\begin{align}
&\frac{\alpha}{1-\alpha}\log F^2\left(\tr_{\tilde{A}_n}\left[ U_{n}(\rho_{RA}^{\ox n} \ox \sigma_{A'_n})U_{n}^\dg \right], \omega_{R^n} \ox \omega_{\bar{A}_n} \right) \nb \\
\leq &D_{\beta}^*(\omega_{R^n\bar{A}_n\tilde{A}_n}\|\omega_{R^n}\ox\tau_{\bar{A}_n\tilde{A}_n})
     -I_{\alpha}^*(R^n:A^n)_{\rho^{\ox n}} \nb \\
=&\bar{I}^{*}_{\beta}(R^n:\bar{A}_n\tilde{A}_n)_\omega-I_{\alpha}^*(R^n:A^n)_{\rho^{\ox n}} \nb \\
\leq &\bar{I}^{*}_{\beta}(R^n:\bar{A}_n)_\omega+2\log|\tilde{A}_n|
      -I_{\alpha}^*(R^n:A^n)_{\rho^{\ox n}} \nb \\
 =   &2\log|\tilde{A}_n|-I_{\alpha}^*(R^n:A^n)_{\rho^{\ox n}}, \label{eq:opt3}
\end{align}
where the fourth line follows from the dimension bound for $\bar{I}^{*}_{\beta}$ (Proposition~\ref{prop:mainpro}~(\romannumeral6)). Optimizing Eq.~\eqref{eq:opt3} over the states $\omega_{R^n}$ and $\omega_{\bar{A}_n}$ followed by a slight reorganization lets us obtain
\beq\label{eq:opt4}
\frac{-1}{n}\log \mathfrak{P}^{\rm dec}(\rho_{RA}^{\ox n},\mc{D}_n)
\geq \frac{1-\alpha}{\alpha}\left\{\frac{1}{n}I_{\alpha}^*(R^n:A^n)_{\rho^{\ox n}}
     -\frac{2}{n}\log|\tilde{A}_n|\right\};
\eeq
note that the performance function $\mathfrak{P}^{\rm dec}$ is defined in Eq.~\eqref{eq:def-p-dec}. Taking the limit, Eq.~\eqref{eq:opt4} translates to
\beq\label{eq:opt5}
\limsup_{n\rar\infty}\frac{-1}{n}\log \mathfrak{P}^{\rm dec}(\rho_{RA}^{\ox n},\mc{D}_n)
\geq \frac{1-\alpha}{\alpha}\left\{I_{\alpha}^{*, \rm{reg}}(R:A)_\rho-2r\right\},
\eeq
where for the right hand side we have also used Eq.~\eqref{eq:opt} and the definition of $I_{\alpha}^{*, \rm{reg}}$ in Eq.~\eqref{eq:defrmi}. Now, Eq.~(\ref{eq:opt5}), Eq.~(\ref{eq:opt}) and the definition of $E_{\rm{sc}}^{\rm{dec}}(\rho_{RA}, r)$ lead to
\begin{equation}
\label{eq:opt6}
E_{\rm{sc}}^{\rm{dec}}(\rho_{RA}, r) \geq \frac{1-\alpha}{\alpha}\left\{I_{\alpha}^{*, \rm{reg}}(R:A)_\rho-2r
\right\}.
\end{equation}
Eq.~(\ref{eq:opt6}) was proved for any $\alpha\in[\frac{1}{2}, 1)$. It also holds for $\alpha=1$ because by definition the left hand side is nonnegative while the right hand side is obviously $0$ in this case. Taking the supremum over $\alpha\in[\frac{1}{2},1]$ we arrive at
\begin{equation}
E_{\rm{sc}}^{\rm{dec}}(\rho_{RA}, r) \geq \sup_{\frac{1}{2} \leq \alpha \leq 1} \frac{1-\alpha}{\alpha}\left\{I_{\alpha}^{*, \rm{reg}}(R:A)_\rho-2r
\right\}.
\end{equation}
\end{proof}

\begin{remark}
Eq.~\eqref{eq:opt5} holds as well if we replace the ``$\limsup$'' with a ``$\liminf$'', which makes the statement stronger. Recall that in the proof of the achievability part, we have upper bounded the strong converse exponent defined with the ``$\limsup$'', which is stronger than that with a ``$\liminf$''. So, similar to the case of privacy amplification, our proof shows that, in the definition of $E_{\rm{sc}}^{\rm{dec}}(\rho_{RA}, r)$ in Eq.~\eqref{eq:defsce}, we can use either ``$\limsup$'' or ``$\liminf$'' to quantify the rate of exponential decay of $\mathfrak{P}^{\rm dec}(\rho_{RA}^{\ox n},\mc{D}_n)$, without changing the result of the strong converse exponent.
\end{remark}

\section{Conclusion and Discussion}
  \label{sec:discussion}
By studying the strong converse exponents for several quantum information tasks, we have found for the first time an operational interpretation to the sandwiched R\'enyi divergence of order $\alpha\in(\frac{1}{2},1)$, as well as to its induced quantum information quantities, namely, the sandwiched R\'enyi conditional entropy and the regularized sandwiched R\'enyi mutual information. We complete the paper with some discussion as follows.

We have adopted a definition of the performance functions, in Eq.~\eqref{eq:def-p-pa} for quantum privacy amplification to also optimize the state of the adversary ($E$), and in Eq.~\eqref{eq:def-p-dec} for quantum information decoupling to also optimize the state of the reference ($R$). Alternatively, one can define the performance functions such that the ideal state on the adversary's system and that on the reference system, are fixed as $\rho_E$ and $\rho_R$, respectively. While for the reliability functions~\cite{LYH2023tight,LiYao2021reliability} these two kinds of definitions yield the same results, for strong converse exponents they are quite different. Our method does not seem to apply to the alternative definition in an obvious way. The main obstruction lies in the achievability parts. Under the alternative definition, the proofs of Theorem~\ref{thm:PA-F} and Theorem~\ref{thm:Fdeco} can be modified to give new bounds, where the testing states $\tau_{XE}$ and $\tau_{RA}$ are restricted such that $\tau_E=\rho_E$ and $\tau_R=\rho_R$. However, these new bounds do not easily translate to the desired results any more. We leave it as an open problem. Besides, our solution to the strong converse exponent of quantum information decoupling is not single-letter. To obtain a single-letter formula remains as another open problem.

Our results are based on the fidelity or purified distance. The strong converse exponents for the tasks considered in the present paper under the trace distance are still unknown. In this case the recent works~\cite{SalzmannDatta2022total,Wilde2022distinguishability,SGC2022strong} have derived for the strong converse exponents lower bounds, which are all in terms of Petz's R{\'e}nyi divergence (of order from $0$ to $1$). This is in contrast to our results in terms of the sandwiched R\'enyi divergence (of order from $1/2$ to $1$). However, usually the exponents under these two measures behave differently and they need distinct methods to solve. So, it is an interesting open problem to find out the strong converse exponents under the trace distance, and especially, to identify the correct R\'enyi divergence in this case.

The duality relations established in~\cite{MDSFT2013on,Beigi2013sandwiched,TBH2014relating, HayashiTomamichel2016correlation} tell that, for $\rho_{ABC}\in\mc{S}(ABC)$ a pure state, $\tau_A\in\mc{P}(A)$ such that $\supp(\rho_A)\subseteq\supp(\tau_A)$, $\alpha\geq\frac{1}{2}$ and $\beta>0$, we have
\begin{align}
\min_{\sigma_B\in\mc{S}(B)}D^*_\alpha(\rho_{AB}\|\tau_A\ox\sigma_B)
&=-\min_{\sigma_C\in\mc{S}(C)}D^*_{\frac{\alpha}{2\alpha-1}}(\rho_{AC}\|\tau_A^{-1}\ox\sigma_C), \label{eq:dual-1} \\
\min_{\sigma_B\in\mc{S}(B)}D_\beta(\rho_{AB}\|\tau_A\ox\sigma_B)
&=-D^*_{\beta^{-1}}(\rho_{AC}\|\tau_A^{-1}\ox\rho_C).
\label{eq:dual-2}
\end{align}
For properly chosen $\tau_A$, the quantities in Eq.~\eqref{eq:dual-1} and Eq.~\eqref{eq:dual-2} reduce to certain kinds of R\'enyi conditional entropy or R\'enyi mutual information. These two identities relate one R\'enyi divergence of order less than $1$ to another R\'enyi divergence of certain order greater than $1$. It is conceivable that their operational interpretations may be similarly related. Indeed, recently Renes~\cite{Renes2022achievable} has employed Eq.~\eqref{eq:dual-2} in the derivation of the error exponents of data compression with quantum side information and of communication over symmetric CQ channels. We think that more findings in this direction can be expected.

\bigskip
\appendix
\section{Auxiliary Lemmas}
\begin{lemma}
\label{lem:appen1}
Let $\rho, \sigma\in\mc{S}(\mc{H})$, and let $\mc{H}=\bigoplus_{i\in \mc{I}}\mc{H}_i$ decompose into a set of mutually orthogonal subspaces $\{\mc{H}_i\}_{i\in \mc{I}}$. Suppose that $\sigma=\sum_{i \in \mc{I}} \sigma_i$ with $\supp(\sigma_i)\subseteq \mc{H}_i$. Then
\begin{equation}
F\big(\sum_{i \in \mc{I}} \Pi_i \rho \Pi_i, \sigma\big) \leq \sqrt{|\mc{I}|} F(\rho, \sigma),
\end{equation}
where $\Pi_i$ is the projection onto $\mc{H}_i$.
\end{lemma}
\begin{proof}
Since the subspaces $\{\mc{H}_i\}_{i\in \mc{I}}$ are mutually orthogonal, by direct calculation we have
\begin{align}
F\big(\sum_{i\in \mc{I}} \Pi_i \rho \Pi_i, \sigma\big)
&=\sum_{i\in \mc{I}} \left\|\sqrt{\Pi_i \rho \Pi_i} \sqrt{\sigma_i}\right\|_1
 =\sum_{i\in \mc{I}} \left\|\sqrt{\rho} \sqrt{\sigma_i}\right\|_1 \nonumber
 =|\mc{I}|\tr\sum_{i\in \mc{I}}\frac{1}{|\mc{I}|}\sqrt{\sqrt{\rho}\sigma_i\sqrt{\rho}} \\
&\leq |\mc{I}|\tr\sqrt{\sqrt{\rho}\big(\sum_{i\in \mc{I}}\frac{1}{|\mc{I}|}\sigma_i\big)\sqrt{\rho}}
 =\sqrt{|\mc{I}|} F(\rho, \sigma),
\end{align}
where for the inequality we use the operator concavity of the function $f(x)=\sqrt{x}$.
\end{proof}

\begin{lemma}
\label{lem:fidelity-re}
Let $\rho, \sigma, \tau \in\mc{S}(\mc{H})$ be any quantum states. Then we have
\begin{equation}\label{eq:fid-re}
-\log F^2(\rho,\sigma) \leq D(\tau\|\rho) + D(\tau\|\sigma).
\end{equation}
\end{lemma}
\begin{proof}
If $\supp(\tau)\nsubseteq \supp(\rho)\cap\supp(\sigma)$ (this includes the case that $\supp(\rho)\cap\supp(\sigma)$=\{0\}), then the right hand side of Eq.~\eqref{eq:fid-re} is $+\infty$ and the statement holds trivially. So, in the following, we assume that $\supp(\tau)\subseteq \supp(\rho)\cap\supp(\sigma)$.

We make use of the fact that the fidelity can be achieved by quantum measurement~\cite{FuchsCaves1995mathematical}. Specifically, there exists a quantum measurement
$\{M_x\}_x$ such that
\beq
F(\rho,\sigma)=\sum_x\sqrt{p_xq_x},
\eeq
where $p_x=\tr\rho M_x$ and $q_x=\tr\sigma M_x$. Let $t_x=\tr\tau M_x$ for all $x$. Then the data processing inequality of the relative entropy gives that
\begin{align}
D(t\|p)\leq D(\tau\|\rho), \\
D(t\|q)\leq D(\tau\|\sigma),
\end{align}
where $t$ is a probability vector with elements $t_x$, and $p$ and $q$ are defined similarly.
So, we have
\begin{align}
     &D(\tau\|\rho)+D(\tau\|\sigma)+\log F^2(\rho, \sigma) \nb\\
\geq &D(t\|p)+D(t\|q)+2\log \sum_x\sqrt{p_xq_x} \nb\\
  =  &2\sum_x t_x\log\frac{t_x}{\sqrt{p_xq_x}} + 2\log \sum_x\sqrt{p_xq_x} \nb\\
  =  &2\sum_x t_x\Big(\log t_x - \log\frac{\sqrt{p_xq_x}}{\sum_x\sqrt{p_xq_x}}\Big) \nb\\
  =  &2D(t\| r) \nb\\
\geq &0, \label{eq:fdeq}
\end{align}
where $r$ is a probability vector with elements $r_x=\frac{\sqrt{p_xq_x}}{\sum_x\sqrt{p_xq_x}}$.
Eq.~\eqref{eq:fdeq} lets us complete the proof.
\end{proof}

\begin{lemma}
\label{lem:srrI}
Let $\rho_{RA}\in\mc{S}(RA)$. Then
\begin{enumerate}[(i)]
  \item the strong converse exponent $E_{\rm{sc}}^{\rm{dec}}(\rho_{RA},r)$ given in Theorem~\ref{theorem:maindec} is strictly positive if and only if $r<\frac{1}{2}I(R:A)_\rho$,
  \item we have $I_{\alpha}^{*, \rm{reg}}(R:A)_\rho\rar I(R:A)_\rho$ when $\alpha\nearrow1$, and
  \item the following equality holds:
  \begin{equation}\label{eq:srrI-2}
 \sup_{\frac{1}{2}\leq\alpha< 1}\frac{1-\alpha}{\alpha}\big\{I_{\alpha}^{*,\rm{reg}}(R:A)_\rho-2r\big\}
=\sup_{\frac{1}{2}\leq\alpha\leq1}\frac{1-\alpha}{\alpha}\big\{I_{\alpha}^{*,\rm{reg}}(R:A)_\rho-2r\big\}.
\end{equation}
\end{enumerate}
\end{lemma}
\begin{proof}
For the state $\rho_{RA}$ and a decoupling scheme $\mc{D}=(\sigma_{A'}, U:\mc{H}_{AA'} \rightarrow \mc{H}_{\bar{A}\tilde{A}})$, We consider another definition of the performance:
\begin{equation}\label{eq:defpalt}
\widetilde{\mathfrak{P}}^{\rm dec}(\rho_{RA},\mc{D})
:=\max_{\omega_{\bar{A}}\in \mc{S}(\bar{A})}
F^2(\tr_{\tilde{A}}\big[U(\rho_{RA} \ox \sigma_{A'})U^\dg  \big], \rho_R \ox \omega_{\bar{A}}).
\end{equation}
Using this performance measure, we define the strong converse exponent in the asymptotic situation as
\begin{equation} \label{eq:defscealt}
\widetilde{E}_{\rm{sc}}^{\rm{dec}}(\rho_{RA},r):=
\inf \left\{\limsup_{n\rar\infty} \frac{-1}{n}\log \widetilde{\mathfrak{P}}^{\rm dec}(\rho_{RA}^{\ox n},\mc{D}_n)~\Big|~\limsup_{n\rar\infty}\frac{1}{n} \log|\tilde{A}_n| \leq r\right\}.
\end{equation}
This is similar to the definition of $E_{\rm{sc}}^{\rm{dec}}(\rho_{RA},r)$ in Eq.~\eqref{eq:defsce}. In~\cite[Theorem 3.8]{LWD2016strong}, three lower bounds were provided for the strong converse exponent of quantum state splitting, which can be translated to quantum information decoupling with the new definitions of Eq.~\eqref{eq:defpalt} and Eq.~\eqref{eq:defscealt}. This translation was shown in~\cite[Proposition 11]{LiYao2021reliability} between quantum state merging and quantum information decoupling. Because quantum state splitting and quantum state merging are inverse processes to each other, the same argument works for the translation between quantum state splitting and quantum information decoupling. In this way, the third bound in~\cite[Theorem 3.8]{LWD2016strong} translates to
\begin{equation}\label{eq:defscealt1}
\widetilde{E}_{\rm{sc}}^{\rm{dec}}(\rho_{RA},r)
\geq\sup_{\frac{1}{2}< \alpha < 1} \frac{1-\alpha}{\alpha}
\big\{\bar{I}^*_{\alpha}(R:A)_{\rho}-2r\big\}.
\end{equation}
We point out that we do not know whether the converse direction of Eq.~\eqref{eq:defscealt1} is correct.

It has been proved in~\cite[Lemma 16]{Sharma2014a} that for any $\rho_{RA} \in \mc{S}(RA)$, $\sigma_A \in \mc{S}(A)$ and $\omega_{R} \in \mc{S}(R)$, we have
\begin{equation}
\label{eq:defscealt2}
F(\rho_{RA}, \rho_R \ox \sigma_A) \geq F^2(\rho_{RA}, \omega_R \ox \sigma_A).
\end{equation}
By definition, Eq.~\eqref{eq:defscealt2} implies that
\begin{equation}
\label{eq:defscealt3}
E_{\rm{sc}}^{\rm{dec}}(\rho_{RA},r)
\geq \frac{1}{2}\widetilde{E}_{\rm{sc}}^{\rm{dec}}(\rho_{RA},r)
\geq \frac{1}{2} \sup_{\frac{1}{2}< \alpha < 1} \frac{1-\alpha}{\alpha}
     \big\{\bar{I}^*_{\alpha}(R:A)_{\rho}-2r\big\},
\end{equation}
where the second inequality is by Eq.~\eqref{eq:defscealt1}. Since $\alpha\mapsto\bar{I}^*_{\alpha}(R:A)_{\rho}$ is continuous, and converges to $I(R:A)_\rho$ as $\alpha$ goes to $1$~\cite{HayashiTomamichel2016correlation},  Eq.~(\ref{eq:defscealt3}) implies that when $r<\frac{1}{2}I(R:A)_\rho$, we have $E_{\rm{sc}}^{\rm{dec}}(\rho_{RA},r)>0$.
Conversely, if $E_{\rm{sc}}^{\rm{dec}}(\rho_{RA},r)>0$, Then due to Theorem~\ref{theorem:maindec}, there exists an $\alpha \in [\frac{1}{2},1)$  such that
\[
r<\frac{1}{2}I_{\alpha}^{*,\rm{reg}}(R:A)_\rho
\leq \frac{1}{2}I_{\alpha}^*(R:A)_\rho
\leq\frac{1}{2}I(R:A)_\rho.
\]
This proves the first statement.

The second statement is a direct result of Theorem~\ref{theorem:maindec} and the first statement. Here note that by definition, it is not difficult to show that the function $\alpha\mapsto I_{\alpha}^*(R^n:A^n)_{\rho^{\ox n}}$ is monotonically increasing for any integer $n$, which further implies that
\beq
\alpha\mapsto I_{\alpha}^{*, \rm{reg}}(R:A)_\rho
\equiv \lim_{n\rar\infty} \frac{1}{n}I_{\alpha}^*(R^n:A^n)_{\rho^{\ox n}}
\eeq
is monotonically increasing, on $[\frac{1}{2},1)$.

The third statement is easy. Since $\alpha\mapsto I_{\alpha}^{*, \rm{reg}}(R:A)_\rho$ is monotonically increasing, including $\alpha=\frac{1}{2}$ in the optimization will not help. On the other hand, the factor $\frac{1-\alpha}{\alpha}$ in the expression is $0$ for both $\alpha=1$ and $\alpha\nearrow1$, this ensures that including $\alpha=1$ in the optimization does not make any difference.
\end{proof}

\medskip
The following lemma is essentially from~\cite{ADHW2009mother}. However, in~\cite{ADHW2009mother} there is only a one-shot decoupling theorem (Theorem 4.2), together with asymptotic results on related quantum tasks. We extract the following statement from~\cite{ADHW2009mother}, and provide a complete proof for reader's convenience.
\begin{lemma}[decoupling theorem~\cite{ADHW2009mother}]
\label{lem:decoupling}
Let $\rho_{RA}\in\mc{S}(RA)$ be a quantum state, and $\epsilon>0$ be any small  constant. There exists a sequence of isometries $U_n:\mc{H}_{A^n}\rar\mc{H}_{\bar{A}_n\tilde{A}_n}$ with $n\in\mathbb{N}$, such that
\begin{align}
&\lim_{n \rar \infty} \left\|\tr_{\tilde{A}_n}\big[U_{n}\rho_{RA}^{\ox n}
   U_{n}^\dg \big]-\rho_R^{\ox n} \ox \omega_{\bar{A}_n} \right\|_1=0, \label{eq:dec-1} \\
&\lim_{n\rar\infty} \frac{1}{n} \log |\tilde{A}_n| = \frac{1}{2}I(R:A)_\rho+\epsilon, \label{eq:dec-2}
\end{align}
where $\omega_{\bar{A}_n}$ is the reduced state of $\tr_{\tilde{A}_n}[U_{n}\rho_{RA}^{\ox n}U_{n}^\dg ]$ on $\bar{A}_n$.
\end{lemma}
\begin{proof}
For any bipartite state $\psi_{RA}\in\mc{S}(RA)$, let $\sigma_{R\bar{A}}(U)=\tr_{\tilde{A}}[(\1_R\ox U)\psi_{RA}(\1_R\ox U)^\dg ]$ be the state remaining on $R\bar{A}$ after the unitary transformation $U$ has been applied to $A=\bar{A}\tilde{A}$.
Theorem~IV.2 of the arXiv version of~\cite{ADHW2009mother} states that
\begin{equation}\label{eq:dec-3}
\int_{\mathbb{U}(A)}\big\|\sigma_{R\bar{A}}(U)-\psi_R\ox\sigma_{\bar{A}}(U)\big\|_1^2\di\,U
\leq \frac{|R||A|}{|\tilde{A}|^2}\left\{\tr[(\psi_{RA})^2]+\tr[(\psi_{R})^2]\tr[(\psi_{A})^2]\right\},
\end{equation}
where $\mathbb{U}(A)$ is the set of all unitary transformations on $A$, and the integration is with respect to the Haar measure on $\mathbb{U}(A)$.

In order to apply the above result in the i.i.d.\ case, we need to consider a modified version of the underlying state by locally restricting it to the typical subspaces. Let $\rho_{RAB}$ be the purification of $\rho_{RA}$. We can define projectors $\Pi_R^n$, $\Pi_A^n$, $\Pi_B^n$ onto the $\delta$-typical subspaces of $\mc{H}_{R^n}$, $\mc{H}_{A^n}$ and $\mc{H}_{B^n}$, respectively. Let $\psi_{R^nA^nB^n}$ be the normalized version of
\[
\Pi_R^n\ox \Pi_A^n\ox \Pi_B^n \rho_{RAB}^{\ox n} \Pi_R^n\ox \Pi_A^n\ox \Pi_B^n.
\]
Appendix A and Section 7 and 8 of~\cite{ADHW2009mother} show that, for a suitable definition of the ``$\delta$-typical subspace'', the following properties hold for any subsystem $X=R, A, B$:
\begin{enumerate}[(i)]
  \item $\tr[\rho_X^{\ox n}\Pi_X^n]\rar 1$, as $n\rar\infty$,
  \item $\|\psi_{R^nA^n}-\rho_{RA}^{\ox n}\|_1\rar 0$, as $n\rar\infty$,
  \item $2^{n[H(X)_\rho-\delta]}\leq\rank(\Pi_X^n)\leq 2^{n[H(X)_\rho+\delta]}$, and
  \item $\tr\left[(\psi_{X^n})^2\right]\leq 2^{-n[H(X)_\rho-\delta]}$.
\end{enumerate}

Now, we can apply the one-shot decoupling theorem to the state $\psi_{R^nA^n}=\tr_{B^n}[\psi_{R^nA^nB^n}]$. Denote the above $\delta$-typical subspace of $\mc{H}_{A^n}$ by $\mc{H}_{A^n}^t$, and its orthogonal complement by $\mc{H}_{A^n}^{at}$. Eq.~\eqref{eq:dec-3} implies that, there exists a unitary transformation $U'_n$ on $\mc{H}_{A^n}^t=\mc{H}_{\bar{A}'_n}\ox\mc{H}_{\tilde{A}'_n}$ such that
\begin{align}
    &\left\|\tr_{\tilde{A}'_n}\left[U'_{n}\psi_{R^nA^n}{U'_n}^{\dg }\right]-
     \psi_{R^n}\ox\tr_{\tilde{A}'_n}\left[U'_{n}\psi_{A^n}{U'_n}^{\dg }\right]\right\|_1 \nb\\
\leq&\frac{\rank(\Pi_R^n)^\frac{1}{2}\,\rank(\Pi_A^n)^\frac{1}{2}}{|\tilde{A}'_n|}
     \left\{\tr[(\psi_{R^nA^n})^2]+\tr[(\psi_{R^n})^2]\tr[(\psi_{A^n})^2]\right\}^\frac{1}{2}.
     \label{eq:dec-4}
\end{align}
We choose
\beq\label{eq:dec-5}
|\tilde{A}'_n|=2^{n[\frac{1}{2}I(R:A)_\rho+3\delta]}.
\eeq
To estimate the right hand side of Eq.~\eqref{eq:dec-4}, we employ the above properties (\romannumeral3) and (\romannumeral4), and make use of the relations $\tr[(\psi_{R^nA^n})^2]=\tr[(\psi_{B^n})^2]$ and $H(R)_\rho+H(A)_\rho\geq H(RA)_\rho=H(B)_\rho$. This lets us obtain
\beq\label{eq:dec-6}
\left\|\tr_{\tilde{A}'_n}\left[U'_{n}\psi_{R^nA^n}{U'_n}^{\dg }\right]-
     \psi_{R^n}\ox\tr_{\tilde{A}'_n}\left[U'_{n}\psi_{A^n}{U'_n}^{\dg }\right]\right\|_1
\leq \sqrt{2}\times2^{-n\delta}\, \stackrel{n\rar\infty}{\lrar} 0.
\eeq
Next, we extend the action of $U'_{n}$ to be on the whole space $\mc{H}_{A^n}$ by letting $U'_{n}\mc{H}
_{A^n}^{at}=0$, such that $U'_{n}$ becomes a partial isometry satisfying ${U'_n}^\dg U'_{n}=\Pi_A^n$. After this extension, Eq.~\eqref{eq:dec-6} still holds because the state $\psi_{A^n}$ is supported only on $\mc{H}_{A^n}^t$. Now due to the above property~(\romannumeral2), we employ the contractivity of the trace norm under partial trace, and the triangle inequality, to derive from Eq.~\eqref{eq:dec-6} that
\beq\label{eq:dec-7}
\left\|\tr_{\tilde{A}'_n}\left[U'_{n}\rho_{RA}^{\ox n}{U'_n}^{\dg }\right]-\rho_{R}^{\ox n}\ox
\tr_{\tilde{A}'_n}\left[U'_{n}\rho_{A}^{\ox n}{U'_n}^{\dg }\right]\right\|_1\rar 0.
\eeq

Eqs.~\eqref{eq:dec-5} and \eqref{eq:dec-7} are already quite close to the statement of this lemma, though $U'_{n}$ is not an isometry yet. In the following, we modify $U'_{n}$ to construct an isometry $U_{n}$ that satisfies Eqs.~\eqref{eq:dec-1} and \eqref{eq:dec-2}. We are able to do so because the weight of the state $\rho_{RA}^{\ox n}$ on the atypical subspace $\mc{H}_{A^n}^{at}$ asymptotically disappears. Let the dimension of the system $\bar{A}_n$ be $|\bar{A}_n|=\max\{|\bar{A}'_n|,\lceil\dim(\mc{H}_{A^n}^{at})/|\tilde{A}'_n|\rceil\}$, such that we can fix two arbitrary isometries $V_n:\mc{H}_{\bar{A}'_n}\rar\mc{H}_{\bar{A}_n}$ and $W_n:\mc{H}_{A^n}^{at}\rar\mc{H}_{\bar{A}_n}\ox\mc{H}_{\tilde{A}'_n}$. Extend $W_n$ to act on the whole space $\mc{H}_{A^n}$ by letting $W_n\mc{H}_{A^n}^{t}=0$, so that now $W_n$ is a partial isometry satisfying $W^\dg _nW_n=\1_{A^n}-\Pi_A^n$. We further introduce a qubit system $\tilde{A}''$ and set $\tilde{A}_n:=\tilde{A}'_n\tilde{A}''$. Then the isometry $U_n:\mc{H}_{A^n}\rar\mc{H}_{\bar{A}_n\tilde{A}_n}$ is constructed as
\beq\label{eq:dec-8}
U_n:=\big(V_n\ox\1_{\tilde{A}'_n}\big)U'_n\ox\ket{0}_{\tilde{A}''}+W_n\ox\ket{1}_{\tilde{A}''}.
\eeq
It is easy to check that $U_n^\dg U_n=\1_{A^n}$ and hence $U_n$ is indeed an isometry. Furthermore, due to our pick of $|\tilde{A}'_n|$ in Eq.~\eqref{eq:dec-5} and $\tilde{A}''$ as a qubit system, we have
\beq\label{eq:dec-9}
\lim_{n\rar\infty} \frac{1}{n} \log |\tilde{A}_n|
= \lim_{n\rar\infty} \frac{1}{n} \left(\log |\tilde{A}'_n|+1\right)
=\frac{1}{2}I(R:A)_\rho+3\delta.
\eeq
Setting $\delta=\epsilon/3$ verifies Eq.~\eqref{eq:dec-2}. The remaining thing is to show Eq.~\eqref{eq:dec-1}. By direct calculation, we get that
\begin{align}
&\ \ \ \,\left\|\tr_{\tilde{A}_n}\left[U_{n}\rho_{RA}^{\ox n}U_n^\dg \right]-V_n\tr_{\tilde{A}'_n}
   \left[U'_{n}\rho_{RA}^{\ox n} {U'_n}^\dg \right]V_n^\dg \right\|_1 \nb\\
&=\left\|\tr_{\tilde{A}'_n}\left[W_n\rho_{RA}^{\ox n}W_n^\dg \right]\right\|_1 \nb\\
&=\tr\left[\rho_{A}^{\ox n}(\1_{A^n}-\Pi_A^n)\right] \nb\\
&\rar 0,    \label{eq:dec-10}
\end{align}
where the last line is by property~(\romannumeral1). Eq.~\eqref{eq:dec-7} and Eq.~\eqref{eq:dec-10} togehter imply Eq.~\eqref{eq:dec-1}.
\end{proof}

\medskip
The following lemma was implicitly proved in~\cite{LWD2016strong}, and it first appeared
in~\cite[Lemma~1]{WangWilde2019resource}. It is a generalization of an special case for $\tau=\1$
of~\cite[Lemma~4.5]{DamHayden2002renyi}.
\begin{lemma}[\cite{LWD2016strong,WangWilde2019resource}]
\label{lem:LWDg}
Let $\rho, \sigma \in \mc{S}(\mc{H})$ and $\tau \in \mc{P}(\mc{H})$, and suppose $\supp(\sigma)\not\perp\supp(\tau)$. Fix $\alpha \in (\frac{1}{2}, 1)$ and $\beta\in(1,+\infty)$
such that $\frac{1}{\alpha}+\frac{1}{\beta}=2$. Then
\begin{equation}
\frac{2\alpha}{1-\alpha} \log F(\rho, \sigma) \leq D^*_\beta(\rho\|\tau)-D^*_\alpha(\sigma\|\tau).
\end{equation}
\end{lemma}

\medskip
The following lemma is from~\cite[Proposition 2.8]{LWD2016strong}, and it can be easily derived from Lemma~\ref{lem:LWDg}.
\begin{lemma}[\cite{LWD2016strong}]
\label{lem:LWD}
Let $\rho_{AB}, \sigma_{AB} \in \mc{S}(AB)$. Fix $\alpha \in (\frac{1}{2}, 1)$ and $\beta\in(1,+\infty)$ such that $\frac{1}{\alpha}+\frac{1}{\beta}=2$. Then
\begin{equation}
\frac{2\alpha}{1-\alpha} \log F(\rho_{AB}, \sigma_{AB})
\leq H_{\alpha}^*(A|B)_\rho-H_{\beta}^*(A|B)_\sigma.
\end{equation}
\end{lemma}

\begin{lemma}[Sion's minimax theorem~\cite{Sion1958general}]
\label{lem:minimax}
Let $X$ be a compact convex set in a topological vector space $V$ and $Y$ be a convex subset of a vector space $W$.
Let $f : X \times Y \rar \mathbb{R}$ be such that
\begin{enumerate}[(i)]
  \item $f(x,\cdot)$ is quasi-concave and upper semi-continuous on $Y$ for each $x \in X$, and
  \item $f(\cdot, y)$ is quasi-convex and lower semi-continuous on $X$ for each $y \in Y$.
\end{enumerate}
Then, we have
\begin{equation}
\label{eq:minimax}
\inf_{x \in X} \sup_{y \in Y} f(x,y)= \sup_{y \in Y} \inf_{x \in X} f(x,y),
\end{equation}
and the infima in Eq.~(\ref{eq:minimax}) can be replaced by minima.
\end{lemma}

\acknowledgments
The authors would like to thank Mario Berta, Masahito Hayashi, Shunlong Luo and Dong Yang for comments. They further thank Hao-Chung Cheng and Marco Tomamichel for helpful discussions on the additivity/nonadditivity of the sandwiched R\'enyi mutual information of Eq.~\eqref{eq:srmi-de}. They are especially grateful to an anonymous referee for valuable comments and suggestions, which have helped them improve the manuscript a lot. The research of KL was supported by the National Natural Science Foundation of China  (No. 61871156, No. 12031004), and the research of YY was supported by the National Natural Science Foundation of China  (No. 61871156, No. 12071099).

%\bibliography{references}
%\bibliographystyle{elsarticle-num}

\end{document}